\documentclass[12pt]{amsart}
\usepackage{amsfonts,amssymb,amsthm,amsmath} %bbold}
\usepackage[latin1]{inputenc}
\usepackage{hyperref}
\usepackage{dsfont}
\usepackage{graphicx}

\newtheorem{lemma}{Lemma}

\newtheorem{theorem}[lemma]{Theorem}

\renewcommand{\Re}{\operatorname{Re}}
\renewcommand{\Im}{\operatorname{Im}}
\newcommand{\NN}{{\mathbb N}}

\newcommand{\RR}{{\mathbb R}}
\newcommand{\CC}{{\mathbb C}}
\newcommand{\ud}{{\textsf d}}
\newcommand{\N}{\mathcal{N}}

\newcommand{\Rspan}{\operatorname{Span}}
\newcommand{\Cspan}{\operatorname{Span}}
\newcommand{\cH}{\mathcal{H}}
\newcommand{\cD}{\mathcal{D}}
\newcommand{\cF}{\mathcal{F}}

\newcommand{\Tr}{\operatorname{Tr}}
\newcommand{\Susy}{\mathrm{SUSY}}

\date{5 August 2011}

\title[Discrete Supersymmetric Matrix Models]{Spectral properties in supersymmetric matrix models}

\author[L.~Boulton]{Lyonell Boulton$^1$}
\address{$^1$Department of Mathematics and Maxwell Institute for Mathematical
Sciences, Heriot-Watt University, Edinburgh
EH14 4AS, United Kingdom}
\email{L.Boulton@hw.ac.uk}

\author[M.~P.~Garcia del Moral]{Maria Pilar Garcia del Moral$^2$}
\address{$^2$Departamento de F{\'\i}sica, Universidad de Oviedo, Avda Calvo Sotelo 18. 33007 Oviedo, Spain}
\email{garciamormaria@uniovi.es}

\author[A.~Restuccia]{Alvaro Restuccia$^3$}
\address{$^3$Departamento de F{\'\i}sica, Universidad Sim{\'o}n Bol{\'\i}var, Apartado 89000, Caracas, Venezuela.}
\address{$^3$ Departamento de F{\'\i}sica, Universidad de Oviedo, Avda Calvo Sotelo 18. 33007 Oviedo, Spain}
\email{arestu@usb.ve}
\begin{document}

\maketitle

%%%%%%%%%%%%%%%%%%%%%%%%%%%%%%%%%%%%%%%%%%%%%%%%%

\begin{abstract}
We formulate a general sufficiency criterion for discreteness of the spectrum of both supersymmmetric and non-su\-persymmetric theories with a fermionic contribution.
This criterion allows an analysis of  Hamiltonians in complete form rather than just their semiclassical limits. In such a framework we examine spectral properties of various (1+0) matrix models. We consider the BMN model of M-theory compactified on a maximally supersymmetric pp-wave background,
different regularizations of the supermembrane with central charges and a
non-supersymmetric model comprising a bound state of $N$ D2 with $m$ D0. While the first two examples have a purely discrete spectrum, the latter has a continuous spectrum with a lower end given in terms of the monopole charge.
\end{abstract}

\tableofcontents

%%%%%%%%%%%%%%%%%%%%%%%%%%%%%%%%%%%%%%%%%%%%%%%%%

\section{Introduction} \label{section1}
Supersymmetric quantum mechanics models have been used in the analysis of non-perturbative String Theory and in the context of M-theory \cite{witten1, townsend}. The $U(N)$ regularization of the 11D supermembrane \cite{bst} was introduced in \cite{hoppe,dwhn, dwmn}. The resulting action is the so called $(1+0)$ SYM theory. In \cite{dwln} it was established that this model has a continuous spectrum comprising the interval $[0,+\infty)$. The interpretation of this theory in terms of D0-branes was introduced in \cite{townsend}. The same $(1+0)$ action was employed in the formulation of the so-called Matrix Model Conjecture \cite{BFSS}. The continuity of the spectrum turns out to be an essential ingredient in this interpretation. The $(1+0)$ SYM action was first obtained in \cite{halpern,flume,rittenberg} in the context of supersymmetric quantum models
unrelated to M-theory. For other nonperturbative bosonic spectral analysis of  inside the context of regularized p-branes see \cite{amilcar3} for the M5-brane case and \cite{gmmnpr} for the regularized ABJ/M constructions for Super Chern-Simons-Matter theories.

Semiclassical analysis is a fundamental tool in the study of physical systems. However, in general, semiclassical limits are not sufficient to characterize many of the properties of the system at high energies. This motivates the use of full nonperturbative approaches. We illustrate this assertion in the context of supersymmetric matrix models. Consider the following 2-dimensional toy Hamiltonian further discussed in
Section~\ref{subsection63}. The example shows that the existence of mass terms does not generally guarantee discreteness of the spectrum beyond the semiclassical approximation. Let
 \begin{equation}\label{matrix222}
H=\begin{pmatrix} -\Delta+V_B(x,y)
 & x+iy+i \\
x-iy-i&-\Delta+V_B(x,y)
\end{pmatrix}
\end{equation}
where $V_B(x,y)=x^2(y+1)^2+y^2$. The model has no flat directions with zero potential and its semiclassical approximation has a discrete spectrum. On the other hand, however, $H$ has a non-empty continuous spectrum comprising the interval $[1,+\infty)$. Moreover, it also has a bound state $\lambda\approx 0.81419$ below the bottom of the essential spectrum. This demonstrates that, contrary to a common believe in SUSY, properties of the semiclassical limit can be substantially different from those of the actual exact theory at high energies.

The spectrum at high energies encodes information about the higher order interacting terms beyond the semiclassical approximation. The first few bound states provide information about the potential in neighborhoods of the origin, while the asymptotic structure of the spectrum at infinity is related to the behavior of the potential at large distances in the configuration space.

A self-adjoint Hamiltonian has a purely discrete spectrum with accumulation at infinity, if and only if its resolvent is compact. Mathematically this is an amenable property as far as the study of the high energy eigenvalues is concerned. On the one hand, this guarantees the existence of a complete set of eigenfunctions, which can be used to decompose the action of the operator in low/high frequency expansions. On the other hand, the study of eigenvalue asymptotics for the resolvent (or the corresponding heat kernel \cite{br}) in the vicinity of the origin, can be carried out by means of the Lidskii's theorem. None of this typically extends, if the Hamiltonian has a non-empty essential spectrum. In general the spectral theorem guarantees that any self-adjoint Hamiltonian with a non-empty essential spectrum can also be ``diagonalized'' in a generic sense. However, apart from a few canonical cases, properties of the corresponding spectral subspaces cannot be determined in a systematic manner.

\subsection{Aims and scopes of the present paper} \label{subsection11}
In Lemma~\ref{nuevo_lemma} below we establish a general sufficiency criterion for discreteness of the spectrum which is applicable to matrix models. This criterion is usable in models satisfying SUSY conditions or otherwise. A similar statement was already found in a more specialized context in \cite{bgmr}. As it turns, the methods of \cite{bgmr} can  be generalized in great manner and they can be implemented in a variety of other contexts.

The main idea behind the approach currently presented, is to  compare the behavior of the bosonic matrix eigenvalues of the theory with those of the fermionic contribution in every direction of the configuration space.
Some matrix models do satisfy the sufficiency conditions below and they automatically have discrete supersymmetric spectrum with finite multiplicity.
Some of them do not satisfy the criterion and in this case other techniques are require for analyzing the spectrum.

Once we have formulated the main mathematical tool in Section~\ref{section2}, we examine three benchmark models of current interest in sections~\ref{section3}-\ref{section5}. The one considered in Section~\ref{section3} corresponds to the discrete light cone quantization of M-theory on the maximal SUSY pp-waves background of D=11 supergravity \cite{bmn}, called the BMN model. We show that it satisfies the conditions of Lemma~\ref{nuevo_lemma}.

The BMN model has also been derived from the supermembrane on a background \cite{dsjvr} and its semiclassical limit has been examined in \cite{shimada}, see also classical solutions in \cite{hoppe2}. Our nonperturbative results show that the spectrum is discrete with finite multiplicity not only around the vacua, but also in the whole positive real line. Accumulation only occurs at infinity and the resolvent is compact. We should highlight that the results established below only cover the case of a finite $N$. As it turns, the bounds we have found diverges in the large $N$ limit. We stress however, that these bounds are not sharp, so the possibility of an extension to the latter case still is not completely excluded in this approach.

In Section~\ref{section4} we examine a model of supermembrane which was already considered in \cite{bgmr}, the supermembrane with central charges \cite{mrt},\cite{mor}. This is a well-defined sector of the full supermembrane theory whose regularized versions (top-down and $SU(N)$) have purely discrete spectrum. The supermembrane was initially thought to play an analogous role to the string in M-theory,  as it was thought to be a fundamental object in the sense that its transverse excitations could eventually be associated to different particles. As a consequence of the results found in \cite{dwln}, beyond the semiclassical approximation analyzed in \cite{stelle}, and due to its spectral properties the 11D supermembrane was considered as a second quantized object, and in this sense only be defined macroscopically. The compactified supermembrane was further studied in \cite{dwpp} and also in \cite{gmr}, showing that the classical instabilities like string-like spikes could not be ruled out simply by means of the compactification process. We showed in \cite{bgmr}, and now we provide additional evidence, that this argumentation does not carry out to the case of the regularized supermembrane with central charge.

Although a rigurous proof is still lacking, we provide additional evidences to support the conjecture that the spectrum of the theory  in the continuum limit would remain discrete. On one hand the bosonic potential in the continuum  \cite{bgmr2} has the same type of quadratic lower bound as the regularized model. Moreover, there is a well-defined convergence of the regularized eigenvalues to the continuum theory eigenvalues in the semiclassical regime. Furthermore, the regularized bound remains finite in the large $N$ limit. If all these assumptions hold true, the supermembrane with central charges could admit an interpretation as a first quantized theory. In sections~\ref{subsection44} and \ref{subsection46},
we illustrate two concrete regularizations of this model in full detail.

In Section~\ref{section5} we examine a matrix model for the bound state of $(ND2,mD0)$ \cite{witten} which does not satisfy the requirements of Lemma~\ref{nuevo_lemma}. In this case we show (Section~\ref{subsection62}) that the corresponding Hamiltonian has a non-empty continuous spectrum. The D2-D0 model is constructed by including monopoles with a characteristic number $m$ in the $D2$ $U(N)$ mode. For any $N$ and $m$ irrespectively of whether they are prime or not, the spectrum of the model is continuous and it is shifted by the monopole contributions.

Section~\ref{section6}, the final section, is devoted to models with non-empty essential spectrum. We show how the variational approach of Section~\ref{section1} is not only of theoretical importance, but also how it is highly relevant in the numerical study of properties of the Hamiltonian even in the presence of a continuous spectrum.  We discuss spectral approximation properties on toy models with non-empty essential spectrum. In particular we consider the embedding of eigenvalues in the continuous spectrum, including numerical estimations singular Weyl sequences and ground wave functions.

\subsection{Background notation} Below we consider the Hilbert space of states of the form $L^2(\RR^N)\otimes \CC^d$ where $N$ and $d$ are large enough. Here $d$ corresponds to the dimension of the Fermi Fock space.
When sufficiently clear from the context, we will only write $L^2 \equiv L^2(\RR^N)\otimes \CC^d$. The corresponding norm in these spaces will be denoted by
$\|\cdot\|$ and the mean integration by
$\langle \cdot \rangle$. In the Euclidean space, we will denote the norm of vectors by
$|\cdot|$ and the inner product will be either left explicit or some times
will be denoted with a single ``dot''. In this notation, $\|u\|^2=\langle |u|^2 \rangle$
for $u\in L^2$. Different symbols, depending on the context, will be employed to denote the variables of the configuration space.
Hamiltonians are self-adjoint operators in this configuration space. They will always be bounded from below, so their domain can be rigorously defined via the classical
Friedrichs extension process.

\section{A variational approach for regularized Hamiltonians} \label{section2}

We firstly consider a general framework which enables a
variational characterization of the spectrum of  Hamiltonians of
matrix model theories,
irrespectively of their supersymmetric properties in the
regularized regime.
The core idea behind this technique has already
been discussed in \cite{bgmr}. It can be regarded as a natural
extension of the  classical result establishing that the spectrum of a
Schr\"odinger operator will be discrete, if the potential term is
bounded from below and it blows up in every direction at infinity,
see \cite[Theorem XIII.16]{rs4}.

Suppose that in $L^2(\mathbb{R}^{N})\otimes \mathbb{C}^{d}$, the
operator realization of the Hamiltonian has the form
\[
   H= P^2 + V(Q), \quad Q\in \mathbb{R}^N
\]
where $V$ is a hermitean  $d\times d$ matrix whose entries are continuous
functions of the configuration variables $Q$.
 Assume additionally that $V(Q)$ is bounded from below by $b$, that is
 \begin{equation} \label{a}
    V(Q)w\cdot w \geq b |w|^2 \qquad\qquad  w\in \mathbb{C}^d
\end{equation}
where $b\in \RR$ is a constant. Then $H$ is bounded from below by $b$ and the spectrum of $H$ does not intersect the interval $(-\infty,b)$.

The following abstract criterion establishes conditions guaranteeing that the spectrum of $H$ is purely discrete.
An alternative proof of Lemma~\ref{criterion} can be found in \cite{bgmr}.

\begin{lemma} \label{criterion} Let $v_k(Q)$ be the
eigenvalues of the $d\times d$ matrix $V(Q)$.
If all $v_{k}(Q)\to +\infty$ as $|Q|\to +\infty$, then the spectrum
of $H$ consists of a set of isolated
eigenvalues of finite multiplicity accumulating
at $+\infty$.
\end{lemma}
\begin{proof}
Without loss of generality we assume that $v_k(Q)\geq 0$, otherwise we
just have to shift
$H$ by a constant in the obvious manner.
The assumptions  imply that $V(Q)$ satisfies (\ref{a}). Since $H$ is bounded from below, one can apply the
Raleigh-Ritz principle to find the eigenvalues below the essential
spectrum. Let
\[
  \lambda_m(H):=\inf \left(\sup_{\Phi\in L}
  \frac{\langle H\Phi,\Phi\rangle}{\|\Phi\|^2}\right)
\]
where the infimum is taken over all $m$-dimensional subspaces
$L$ of the domain of $H$.  Note that the quotient on the right hand side
is a Rayleigh quotient in the Hilbert space $L^2(\RR^N)\otimes \CC^d$.   Then the bottom of the
essential spectrum of $H$ is $\lim_{m\to +\infty}\lambda_m(H)$. If this limit is $+\infty$, then the spectrum of $H$ is
discrete.

The hypothesis of the lemma is equivalent to the following
condition: for all  $\tilde{b}>0$, there exists a ball $\mathcal{B}\subset \mathbb{R}^N$ (of possibly
very large radius) such that
\[
 V(Q)w\cdot w \geq \tilde{b}|w|^2,\qquad \mathrm{all}\ Q\not \in \mathcal{B}.
\]

Let
\[
 W(Q):=\left\{ \begin{array}{ll} -\tilde{b} & Q\in \mathcal{B} \\ 0 & Q
 \not \in \mathcal{B} \end{array} \right. .
\]

Then for all $\Phi$ in the domain of $H$,
\[
  V(Q)\Phi(Q)\cdot \Phi(Q) \geq \tilde{b}\Phi(Q) \cdot \Phi(Q) +
  W(Q)\Phi(Q)\cdot \Phi(Q)
\]
for almost all $Q\in \RR^d$,
so that
\[
  \langle V\Phi,\Phi\rangle \geq \langle(\tilde{b}+W)\Phi,\Phi\rangle.
\]
Thus
\[
 \lambda_m(H) \geq \tilde{b}+ \lambda_m(P^2 + W)
\]
for all $m=1,2,\ldots$.

Since $W(Q)$ is a bounded potential with compact support, by
Weyl's theorem, the essential spectrum of $P^2+W(Q)$ is
$[0,+\infty)$. Therefore by Raleigh-Ritz criterion, there exists
$\tilde{M}>0$ such that
\[
   \lambda_m(P^2+W(Q)) \geq -1
\]
for all $m>\tilde{M}$. Thus, $\lambda_m(H)\geq \tilde{b}-1$ for all $m$ large enough.
Since we can take $\tilde{b}$
very large, necessarily $\lambda_m\to +\infty$ as $m$ increases.
\end{proof}

As a consequence of this lemma, we can establish the main mathematical contribution of the present paper.

\begin{lemma} \label{nuevo_lemma}
Let $V_B$ be a continuous bosonic potential of the configuration space.
Let $V_F$ be a fermionic matrix potential with continuous entries $v_{ij}$
of the configuration space.
Suppose that there exist constants
$b_B,b_F,R_0,p_B,p_F>0$ independent of $Q\in\RR^N$
satisfying the following conditions
\[
   V_B \geq b_B |Q|^{p_B} \qquad \text{and} \qquad  |v_{ij}|\leq b_F |Q|^{p_F}
\]
for all $|Q|>R_0$. If $p_B>p_F$, then the Hamiltonian
$H=P^2+ V(Q)$ of the quantum system associated to $V=V_BI+V_F$
has spectrum consisting exclusively
of isolated eigenvalues of finite multiplicity, semi-bounded below and accumulating at $+\infty$.\end{lemma}

\begin{proof}
We apply Lemma~\ref{criterion}.  The eigenvalues of the matrix $V_{B}(Q)+V_F(Q)$
are the $\lambda=\lambda(Q)\in \mathbb{R}$ such that
\begin{equation*}
\det [(\lambda-V_{B}(Q))I-V_F(Q)]=0.
\end{equation*}
Let $\widehat{\lambda}(Q)$  be any of the eigenvalues of
$\frac{V_F(Q)}{|Q|^{p_F}}$. The imposed hypothesis ensures that
$|\widehat{\lambda}(Q)|$
remains bounded for all $|Q|\geq R_0$. As
\begin{equation*}
\lambda(Q)=V_{B}(Q)+|Q|^{p_F}\widehat{\lambda}(Q)\ge b_B |Q|^{p_B}-
|\widehat{\lambda}(Q)||Q|^{p_F},
\end{equation*}
for $p_B>p_F$, $\lambda(Q)\to +\infty$ whenever $|Q|\to+\infty$. Note
that $V$ is continuous, hence it is automatically bounded from below.
\end{proof}

For further applications of these results see \cite{bgmr2}.
Also note that
the continuity assumptions in the above
for the potentials may be relaxed to measurable and bounded from below.
Nevertheless the stronger condition of
on continuity will serve our purposes below.

%%%%%%%%%%%%%%%%%%%%%%%%%%%%%%%%%%%%%%%%%%%%%%%%%%%%%%%%%%%%%%%%%%
%%%%%%%%%%%%%%%%%%%%%%%%%%%%%%%%%%%%%%%%%%%%%%%%%%%%%%%%

\section{The BMN matrix model} \label{section3}
The  matrix model for the Discrete Light Cone Quantization (DLCQ)
of M-theory on the maximally supersymmetric pp-waves background of eleven dimensional supergravity examined in \cite{bmn}\footnote{We thank J.~Maldacena for clarifying some details related to the construction of this model.}
fits in well with the abstract framework of the previous section. This is a well-known model in the literature and it admits an interpretation in terms of coincident gravitons. Extensions have been intensively studied also in spaces with less supersymmetry. See also \cite{yolanda}.

The dynamics of this theory is described by an $U(N)$ matrix model, which in our notation reads
\begin{equation}
\begin{aligned}
&L_{BMN}=T-V_B-V_F\\
&V_B=\Tr \left[
\frac{\mu^2}{36R}\sum_{i=1,2,3}(X^i)^2+\frac{\mu^2}{144R}\sum_{i=4}^9
(X^i)^2+\frac{i\mu}{3}\sum_{i,j,k=1}^3\epsilon_{ijk}X^iX^jX^k-\frac{R}{2}\sum_{i,j=1}^9[X^i,X^j]^2\right]\\
\nonumber
&V_F=\Tr \left[\frac{\mu}{4}\Psi^T\gamma_{123}\Psi-2iR\sum_{i=1}^{9}(\Psi^T\gamma_i[\Psi,X^i])\right]
\end{aligned}\end{equation}
The spinorial fields are represented by hermitean matrices in the Fermi Fock space.

To characterize the spectrum of the Hamiltonian
\[H=R \Tr\left[\frac{1}{2}P_i^2+\frac{1}{R}(V_B+V_F)\right],\]
 we split the bosonic potential as
\begin{equation}
\begin{aligned}
&V_{B1}=\Tr\left[\frac{\mu^2}{36R}\sum_{i=1,2,3}(X^i)^2
+\frac{i\mu}{3}\sum_{i,j,k=1}^3\epsilon_{ijk}X^iX^jX^k-\frac{R}{2}\sum_{i,j=1}^3[X^i,X^j]^2\right]\\
\nonumber &V_{B2}=\Tr\left[\frac{\mu^2}{144R}\sum_{i=4}^9
(X^i)^2-\frac{R}{2}\sum_{i,j=4}^9[X^i,X^j]^2 \right]\\ \nonumber
&V_{B3}=\Tr\left[-R\sum_{i=1,2,3;j=4,\dots,9}[X^i,X^j]^2\right]
\end{aligned}\end{equation}
The quartic contribution to the potential
 with an overall minus sign is positive, since the
  commutator is antihermitean. The coordinates
    $X^i$, for $i=4,\dots,9$, only contribute quadratically and quartically to the Lagrangian through the potentials $V_{B2}+V_{B3}$.
Therefore, they satisfy the bound of Lemma~\ref{nuevo_lemma}, with $p_B=2$ and $p_F=1$. Thus, the analysis of the
    bosonic potential may be focus in the first three coordinates.

    Let us concentrate on the $V_{B1}$ contribution.
This potential may be re-written as a perfect square,
\[
V_{B1}=\frac{1}{2}\Tr\left[ -i\sqrt{R}[X^{i},X^{j}]-\frac{\mu}{6\sqrt{R}}\epsilon^{ijk}X^k\right]^2,
\]
so it is positive definite $V_{B1}\ge 0$. It vanishes at the variety determined by the condition
\begin{equation}\label{solo}
[X^i,X^j]=\frac{i\mu}{6R}\epsilon^{ijk}X^k.
\end{equation}
In turns, this condition corresponds to a fuzzy sphere, \cite{bmn}, along the
directions $1,\,2$ and $3$, so there are no flat directions with zero
potential.

Let us now examine the potential away from the minimal set
in the configuration space. To characterize completely the system let  $\rho^2=\sum_{i=1}^3
Tr(X^{i})^2$  and $\varphi\equiv \frac{X}{\rho}$ be defined on a unitary
hypersphere $S^{3N^2-1}$. Let $V_{B1}=\frac{\mu^2}{36R}\rho^2 P(\rho,\varphi)$
where the polynomial
\[
P(\rho,\varphi)=1+2b\rho+a^2\rho^2= \left(a\rho+\frac{b}{a}\right)^2+\left(1-\left(\frac{b}{a}\right)^2\right)
\]
for
\begin{equation} \begin{aligned} b=\Tr \left(
\frac{i3R}{\mu}[\varphi^i,\varphi^j]\varphi^k\epsilon^{ijk}\right)\qquad\text{and}\qquad
\nonumber  a^2=-\frac{18R^2}{\mu^2}\Tr[\varphi^i,\varphi^j].
\end{aligned}
\end{equation}
Since $\frac{b}{|a|}$ is the inner product of two unitary vectors,
 $\vert b\vert\leq |a|$. Consequently,
\[1-\left(\frac{b}{a}\right)^2\geq 0.\]

\begin{theorem} \label{lemma_D0D2}
Let $R_0>\frac{\mu}{3R}\sqrt{C_2(N)N}$ where $C_2(N)=\frac{N^2-1}{4}$ and $\mu,R$ are different from zero.
Then $P(\rho,\varphi)>C>0$ for all $\rho>R_0$ and $\varphi\in S^{3N^2-1}$.
\end{theorem}
\begin{proof}
Assume that there exists a sequence $(\rho_l,\varphi_l)$ such
that $P(\rho_l,\varphi_l)\to 0$ and $\rho_l>R_0$ for large enough $l$. Since $\varphi$ takes the value on
a compact set, there exists a subsequence such that
$\varphi_{li}\to\hat{\varphi}$ when $l_i\to+\infty$. Suppose without loss of generality that $(\rho_l,\varphi_l)$ is such a subsequence. From the expression for $P(\rho,\varphi)$ it follows that
\begin{equation*}
\begin{aligned}
&\left(a(\varphi_l)\rho_l
+\frac{b(\varphi_l)}{a(\varphi_l)}\right)^2\to 0 \qquad \text{and} \qquad
\left(1-\frac{b^2(\varphi_l)}{a^2(\varphi_l)}\right)\to 0.
\end{aligned}
\end{equation*}
Since $\frac{b^2}{a^2}$ is continuous in
$\varphi$ as it is the inner product of two unitary vectors which are in turns continuous in $\varphi$,  we must necessarily have
\begin{equation*}\label{lim}
\left(\frac{b(\hat{\varphi})}{a(\hat{\varphi})}\right)^2=1.
\end{equation*}
This latter property only holds, if the two vectors involved are parallel but point towards opposite directions.
To see this, note that $\rho_l\to-\frac{b}{a^2}$, so that $b$ must be negative. Then
\[
\frac{6iR}{ a(\hat{\varphi})\mu}[\hat{\varphi}^i,\hat{\varphi}^j]=-\epsilon_{ijk}\hat{\varphi}^k.
\]
If we now compare with \eqref{solo}, note that here (by definition)
\begin{equation}\label{norm}\{\hat{\varphi}^k: \Tr(\hat{\varphi}^i)^2=1\}\end{equation}
and there is an extra factor $ a(\hat{\varphi})$.

Write $\hat{\varphi}=CJ^i$ where $C=\frac{a(\hat{\varphi})\mu}{6R}$. Then we get that $J^i$  are associated with the $SU(2)$ algebra satisfying $[J^i,J^j]=i\epsilon_{ijk}J^k$.
By virtue of (\ref{norm})
\[
1=\Tr(\hat{\varphi}^i)^2 =C^2 \Tr(J^i)^2 =C^2 C_2(N) N.
\]
Here $C_2(N)$ is the Casimir invariant. Then,
\[
a(\hat{\varphi})^2=\left(\frac{6R}{\mu}\right)^2\frac{1}{C_2(N)N}.
\]
 Since $\rho_l> R_0$ where
\[
R_0>\frac{2}{a(\hat{\varphi}_l)}=\frac{\mu\sqrt{C_2(N)N}}{3R}=\frac{\mu\sqrt{N^3-N}}{6R},\]
for $l$ large enough, we get
\begin{align*}
 P(\rho_l,\varphi_l) &\geq
a^2(\varphi_l)\left(
\rho_l+\frac{b(\varphi_l)}{a^2(\varphi_l)}\right)^2  \\
& \geq a^2(\varphi_l)\left(
R_0+\frac{b(\varphi_l)}{a^2(\varphi_l)}\right)^2 \\
& \geq 1.
\end{align*}
The latter is clearly a contradiction, so the validity of the theorem is ensured.
\end{proof}

In order to consider the supersymmetric contribution, we just have to realize that the fermionic contribution
 is linear in the bosonic variables, so it satisfies the assumptions of
 Lemma~\ref{nuevo_lemma}. Consequently the
 supersymmetric Hamiltonian of the BMN matrix model has a purely discrete spectrum with eigenvalues of finite multiplicity only accumulating at infinity. As we will see in the next section, this property is also shared with the
 supermembrane with central charges.

 We emphasize here that the spectrum is discrete  in the whole real line.
It should be noted however that, at present, there are not clear restrictions about the spectrum of the model in the large $N$ limit, as $R_0 \to+\infty$ when $N\to+\infty$. In principle, it might have a complicated continuous spectrum with the presence of gaps. The proof of the existence of a gap in this regime remains an interesting open question.

An important advantage of models with discrete spectrum over those with a non-empty essential spectrum, lies in the fact that the behavior of the Hamiltonian at high energies may be determine with precision from the heat kernel or the resolvent operator, by examining accumulation of the spectrum around the origin. These considerations will be discussed elsewhere.

%%%%%%%%%%%%%%%%%%%%%%%%%%%%%%%%%%%%%%%%%%%%%%%%%%%%%%%%%

\section{The supermembrane with central charges} \label{section4}
 The action of the supermembrane with central charges \cite{bgmr2}, with
base manifold a compact Riemann surface $\Sigma$ and Target Space $\Omega$
the product of a compact manifold and a Minkowski space-time, is
defined in term of maps: $\Sigma \longrightarrow \Omega$,
satisfying a certain topological restriction over $\Sigma$. This
restriction ensures that the corresponding maps are wrapped in a
canonical (irreducible) manner around the compact sector of
$\Omega$. In order to generate a nontrivial family of admissible maps,
this sector is not arbitrary but rather it is constrained by the
existence of a holomorphic immersion $\Sigma\longrightarrow
\Omega$.

In particular, let $\Sigma$ be a torus and
$\Omega=T^2\times M_9$ where $T^2=S^1\times S^1$ is the flat
torus. Let $X_r:\Sigma \longrightarrow T^2$ with $r=1,2$ and
$X^m:\Sigma\longrightarrow M_9$ with $m=3,\ldots,9.$ The
topological restriction is explicitly given in this case by the
condition
\begin{equation} \label{winding}
\epsilon_{rs}\int_{\Sigma}\ud X_r\wedge \ud X_s= n \operatorname{Area}(\Sigma)\ne
0.
\end{equation}
Note that the raising and lowering indices of the $X_r$ fields is
consistent  with the $\delta_{rs}$ metric of the Target Space. Here
$\epsilon$ is the Levi-Civita symbol and $n\in\mathbb{N}$ is a fix
constant of the
 model which corresponds to the winding number of the maps.
 The associated holomorphic immersion is defined in terms of the holomorphic one-form over $\Sigma$,
 which may be constructed in terms of a basis of harmonic one-forms over $\Sigma$ denoted by $\ud \widehat{X}_r$.
  The one-forms $\ud X_r$ satisfying \eqref{winding} for  $n=1$ always admit a decomposition
\[
\ud X_r=\ud \widehat{X}_r+\ud A_r
\]
where $\ud A_r$ are exact one-forms.
These one-forms are defined modulo constants on $\Sigma$. In turns
the degrees of
freedom within the sector are realized in terms of single-valued
fields $A_r$ over $\Sigma$.

The functional defining the action of the supermembrane with central
charges in the Light Cone Gauge (LCG)
is analogous to the corresponding functional considered in \cite{dwhn},
but defined in the target space $\Omega$
described above  and restricted
 by the topological constraint \eqref{winding}. See for example \cite{bellorin},\cite{joselen}.
 The corresponding Hamiltonian
 realizes as a Schr{\"o}dinger operator acting  on a dense domain of
$L^2(\RR^N)\otimes \CC^d$ for $N$ and $d$ sufficiently large and
the configuration space coordinates have a representation
$Q\in L^2(\Sigma)$ in terms of harmonic functions.

%%%%%%%%%%%%%%%%%%%%%%%%%%%%%%%%%%%%%%%%%%%%%%%%%%%

\subsection{A top-down truncation: the regularized states space} \label{subsection41}
The relevance of the regularized models we discuss in this section lies in the fact that their nonlinear bosonic potential converges in the $L^2$
norm to the bosonic potential of the supermembrane with central charges, when the dimension $2d$ of $\Rspan\{Y^A:{A\in \N}\}$  (see the definition below)
increases. We are not focusing on the symmetries of the regularized model but in its relation with the 1+2 field theory. These all have discrete spectrum with finite multiplicity and accumulation point at infinity for any $d$ (i.e. there is \underline{not} a continuous sector in the spectrum). To achieve this, we show that the bosonic potential satisfies the bound of Lemma~\ref{nuevo_lemma}
 for $p=2$, since the fermionic potential is linear on the configuration variables. Remarkably, for any regularized version of the fermionic potential with that property, the Hamiltonian will always have a purely discrete spectrum.

Let $\Sigma$ be as above and consider that $\ud \hat{X}_r$ for $r=1,2$ is a basis of harmonic
one-forms in $\Sigma$. Let
\[
     D_r=\frac{\epsilon^{ab}}{\sqrt{W}}(\partial_a\hat{X}_r)\partial_b
\]
where $\partial_a$ denotes differentiation with respect to $\sigma^a$
 and $W$ is the determinant of the induced metric defined by the minimal immersion $\widehat{X}_r$.
 We will denote the symplectic bracket associated to the
supermembrane with central charges by
\[
   \{B,C\}\equiv  \frac{\epsilon^{ab}}{\sqrt{W}} \partial_a B\partial_b C=
   \epsilon^{rs}  D_r B D_s C.
\]
 Let
\[
    \langle B \rangle = \int_\Sigma  \left(\sqrt{W}\ud^2 \sigma\right) B
    = \int_\Sigma \left(\sqrt{W} \ud \sigma^1\wedge \ud \sigma^2\right) B
\]
for integrable fields $B$ of $\Sigma$. In this section $L^2(\Sigma)$ will be the Hilbert
space of square integrable fields in $\Sigma$ with norm given by
$\|B\|=\langle |B|^2 \rangle^{1/2}$.

Let $(Y^A)$ be an orthonormal basis of eigenfunctions of the
Laplacian acting on $L^2(\Sigma)$, where $\overline{Y^A}=Y^{-A}$.
Let $\N\subset \NN\times \NN$ be a finite set of positive
bi-indices. Denote by $\Rspan\{Y^A:{A\in \N}\}$ the
finite-dimensional subspace of real scalar fields generated
by the modes identified with $\N$,
\begin{align*}
   \Rspan\{Y^A\!:\!A\!\in\! \N\}=\Big\{&\sum_{A\in -\N\cup\N} B_A Y^A:
   \\
   &  B_A\in \CC \text{ and } B_{-A}=\overline{B_A} \text{ for all } A\in\N\Big\}.
\end{align*}
Then $\Rspan\{Y^A:A\in \N\} \simeq\RR^{2d}$
in the sense of real linear spaces, via the identification
\[
     \sum B_A Y^A\longmapsto (\Re(B_A),\Im(B_A))_{A\in -\N\cup\N}\in \RR^{2d}.
\]
The inverse of this map is given by
\[
     (B_A,C_A)_{A\in -\N\cup\N}\longmapsto \sum_{A\in -\N\cup\N} (B_A+iC_A)Y^A+
    \sum_{A\in -\N\cup\N} (B_A-iC_A)Y^{-A}.
\]

Assume for simplicity that $n=1$ and $\operatorname{Area}(\Sigma)=1$ in
\eqref{winding}. We may choose $(Y^A)$ to be
\[
    Y^{(n_1,n_2)}(\hat{X}_1,\hat{X}_2)=e^{i(n_1\hat{X}_1+n_2\hat{X}_2)}
\]
where $\hat{X}_1$ and $\hat{X}_2$ are identified as the angles of
the sector  $T^2$ of the Target. The invariance of the action of the
supermembrane with central charges in the LCG under the
area preserving diffeomorphisms allows to fix a gauge {\cite{gmr}. Under this gauge fixing,
the expression for $A_1$ and $A_2$ are the following  (in the finite
dimensional case we are considering $A_r,X^m\in\Cspan\{Y^A:A\in
\N\}$):
\begin{align*}
   A_1&=\sum_{(n_1,0)\in -\N\cup\N} A_1^{(n_1,0)}Y^{(n_1,0)} \tag{g1} \\
   A_2&=\sum_{(n_1,n_2)\in-\N\cup\N,\,n_1\not=0}A_2^{(n_1,n_2)}Y^{(n_1,n_2)}. \tag{g2}
\end{align*}
The definition of $D_r$ yields
\[
    D_rB D_r\overline{B}=\frac{\epsilon^{ac}}{\sqrt{W}}
\partial_a\hat{X}^r \partial_c B\frac{\epsilon^{bd}}{\sqrt{W}}
\partial_b\hat{X}^r\partial_d \overline{B}=g^{ab}\partial_aB \partial_b\overline{B}
\]
and
\[
  \{B,C\}= \frac{\epsilon^{ab}}{\sqrt{W}}\partial_a B \partial_b C=
   \epsilon^{rs}D_r B D_s C
\]
where $g^{ab}$ is the inverse of the metric $g_{ab}$ induced over
$\Sigma$ by the minimal immersion $\hat{X}_r$ introduced above.

We
now define the space of admissible fields in terms of which
the potential component of the bosonic Hamiltonian of the
model is realized explicitly. Let the semi-norm
$\|B\|_1^2=\langle g^{ab}\partial_a B \partial_b \overline{B}
\rangle$ defined on fields for which the integral on the right hand
side is finite. Note that
\[
    \|B\|_1^2=\langle D_r B D_r\overline{B}\rangle =\|D_rB \|^2.
\]
In the local coordinates, the Jacobian
$J=\epsilon^{ab}\partial_a\hat{X}_1\partial_b\hat{X}_2=\sqrt{W}$ and it is different from zero on any point of $\Sigma$.
This ensures that $\|\cdot\|_1$ defines a semi-norm in the Sobolev
space $H^1(\Sigma)$. It clearly is not a norm, since $\|B\|_1=0$ for
any locally constant field $B$.

Denote by $\cH_1(\Sigma)$
the orthogonal complement in $H^1(\Sigma)$ of the space $\mathcal{C}$ generated by all constant fields in $\Sigma$, that is
\[
\cH^1(\Sigma)=H^1(\Sigma)\ominus \mathcal{C}.
\]
Then $\|\cdot\|_1$ induces a norm in $\cH^1(\Sigma)$ and makes it a
Hilbert space. The  fields $A^r$ belong to the exact part in the
decomposition of $\ud X^r$, while constant functions are harmonic
and are contained in the orthogonal subspace. By virtue of (g1) and
(g2),
\begin{align*}
   &D_1A_1=0 \tag{a} \\
   &D_1A_2=0 \quad \Rightarrow \quad A_2=0 \tag{b}.
\end{align*}
Thus
\[
 \|A_2\|_1^2= \langle D_1 A_2D_1A_2\rangle.
\]

%%%%%%%%%%%%%%%%%%%%%%%%%%%%%%%%%%%%%%%%%%%%%%%%%%

\subsection{Discreteness of the top-down regularization} \label{subsection42}
The potential of the bosonic Hamiltonian of the supermembrane
with central charge
is
\[
    V(A_r,X^m)=\frac{1}{4}\left\langle 2\cD_rX^m \cD_r X^m+\cF_{rs}\cF_{rs}+\{X^m,X^n\}^2
    \right\rangle .
\]
 Under the ansatz
(g1)-(g2),
\[
    \langle D_2 A_1 D_1 A_2\rangle =0 \quad \text{and} \quad
    \langle D_2A_1\{A_1,A_2\}\rangle =0,
\]
we get
\begin{equation} \label{potential}
     V(A_r,X^m)=\rho^2+2B+A^2
\end{equation}
where
\begin{align*}
    B(A_r,X^m)=B&=\langle D_r X^m\{A_r,X^m\}+D_1A_2\{A_1,A_2\}\rangle \\
    A(A_r,X^m)=A&=\langle \{A_1,X^m\}^2+\{A_2,X^m\}^2+\{A_1,A_2\}^2+
    \{X^m,X^n\}^2\rangle^{1/2} \\
    \rho(A_r,X^m)=\rho&=\left(\|X^m\|_1^2+\|A_r\|_1^2\right)^{1/2}.
\end{align*}
 Let
\begin{align*}
    W_1&=(D_1A_2,D_2A_1,D_1X^m,D_2X^m,0)\\
    W_2&=(\{A_1,A_2\},0,\{A_1,X^m\},\{A_2,X^m\},\{X^m,X^n\}) .
\end{align*}
Then
\begin{align*}
   \rho^2&=\langle W_1,W_1\rangle= \|W_1\|^2\\
       A^2&=\langle W_2,W_2\rangle = \|W_2\|^2 \\
   B&=\langle W_1,W_2\rangle.
\end{align*}

\begin{lemma} \label{polar}
The functionals $A$, $B$ and $\rho$ are continuous in the norm $\|\cdot\|_1$
of $\cH^1$. Moreover $A$ is homogeneous of order 2 and $B$ is homogeneous of
order 3 in the sense that
\[
   A(cA_r,cX^m)=c^2A(A_r,X^m) \qquad \text{and} \qquad
   B(cA_r,cX^m)=c^3B(A_r,X^m)
\]
for any constant $c\in \RR$.
\end{lemma}
\begin{proof}
The proof is elementary.
For the first part we just need to observe that the bracket $\{\cdot,\cdot \}$
is bi-continuous in the norm of $\cH^1$. For the second part we just need
to observe that
\[
   \{cu,cw\}=c^2\{u,w\}.
\]
\end{proof}

Below we will consider two functionals
\begin{equation*}
   a=\frac{A}{\rho^2} \qquad \text{and} \qquad
   b=\frac{B}{\rho^3}
\end{equation*}
defined for all $(A_r,X^m)\not=0$. By virtue of Lemma~\ref{polar},
they are both continuous in the norm of $\cH^1$ and only depend on
the direction of the field $(A_r,X^m)$, not on $\rho$. In the proof
of Theorem~\ref{bound_V_finite_case} we will make use of this fact.

\begin{lemma} \label{harmonic}
Let $F\in\Cspan\{Y^A:A\in
\N\}$. If $F$ vanish in a subset of $\Sigma$ of positive measure, then it should vanish
identically in the whole of $\Sigma$.
\end{lemma}
\begin{proof}
The statement is a direct consequence of the fact that $F$ is a linear combination of finitely many harmonic functions.
\end{proof}

\begin{theorem} \label{bound_V_finite_case}
There exist a constant $0<k<1$ such that
\[
     V(A_r,X^m)\geq k \rho^2(A_r,X^m)
\]
for all $A_r,X^m\in\Rspan\{Y^A:A\in
\N\}$.
\end{theorem}
\begin{proof}
Since $\Rspan\{Y^A:{A\in \N}\}\simeq \RR^{2d}$ and all norms
in a finite-dimensional subspace are equivalent, $\|\cdot\|_1$
is equivalent to the Euclidean norm. The functionals $A$, $B$ and $\rho$ can then be identified with real-valued functions
 of\footnote{ In the supermembrane the number of degree of freedom is $8$ hence we get $\RR^{16d}$. Evidently the proof holds for any $D>2d$.} $\RR^{D}$.
 Let $S^{D-1}$ be the hypersphere $\rho=1$.
Let $a$ and $b$ be as above.
Since they are independent of $\rho$, we will often denote $a(\phi)\equiv a$ and $b(\phi)\equiv b$ for $\phi\in S^{D-1}$.

Let
\[
    P_\phi(\rho)=1+2b \rho +a^2 \rho^2.
\]
It is readily seen that $V(A_r,X^m)=\rho^2 P_\phi(\rho)$. We will achieve the proof
of the above identity for the potential, by showing that the polynomial
$P_\phi(\rho)$ is bounded uniformly from below for all
$A_r,X^m\in\Rspan\{Y^A\}_{A\in \N}$. Verifying the later requires
a number of steps.

Firstly note that, by the Cauchy-Schwartz inequality,
\begin{equation} \label{C-S}
   |b(\phi)|=\frac{|\langle W_1,W_2\rangle|}{\rho^3}\leq
   \frac{\|W_1\|_{L^2(\Sigma)}\|W_2\|_{L^2(\Sigma)}}{\rho^3}=a(\phi)
   \text{ for all } \phi\in S^{D-1}.
\end{equation}
If $a(\phi)=0$, then also $b(\phi)=0$ and $P_\phi(\rho)\equiv 1$ for any finite $\rho$.
If $a(\phi)\not=0$, then
\[
    P_\phi(\rho)=\left(a\rho+\frac ba\right)^2+\left(1-\frac{b^2}{a^2}\right)
\]
and we are confronted with two further possibilities. If $a(\phi)\not=0$
and $b(\phi)\geq 0$, then
\begin{equation*}
   \min_{\rho\geq 0}P_\phi(\rho)=P_\phi(0)=1.
\end{equation*}
If, on the other hand, $a(\phi)\not=0$ and $b(\phi)< 0$, then
\begin{equation*}
\min_{\rho\geq 0}P_\phi(\rho)=P_\phi(\tilde{\rho})=1-\frac{b^2}{a^2}
\qquad \text{where} \quad \tilde{\rho}=-\frac b{a^2}.
\end{equation*}
In order to complete the proof of the lemma we show that if
$\phi_j\in S^{D-1}$ is a sequence such that $a(\phi_j)\not=0$ and
\[
      \frac{b^2(\phi_j)}{a^2(\phi_j)}\to \ell \qquad\qquad \text{ as } j\to +\infty,
\]
then necessarily $\ell<1$.

By \eqref{C-S}, $\ell \leq 1$. Assume that there exists a sequence
$\hat{\phi}_j\in S^{D-1}$ such that $a(\hat{\phi}_j)\not=0$ and
\begin{equation} \label{inposs}
\frac{b^2(\hat{\phi}_j)}{a^2(\hat{\phi}_j)}\to 1 \qquad \qquad \text{as } j\to +\infty.
\end{equation}
Below we prove that this always lead to a contradiction.
Since $S^{D-1}$ is compact, after extracting a subsequence if necessary,
we can assume that $\hat{\phi}_j\to \hat{\phi}\in S^{D-1}$. There are now two possible cases: either
\begin{gather}
   a(\hat{\phi})\not= 0 \label{poss1} \qquad \text{ or} \\
 a(\hat{\phi})= 0. \label{poss2}
  \end{gather}
We will denote by $(\cdot)_{j}$ and $(\cdot)_{\wedge}$ the evaluation
of the corresponding functionals at (or extracting the corresponding field coordinate
of) $\hat{\phi}_j$ and $\hat{\phi}$ respectively. By continuity,
\begin{gather*}
    \left(\frac{D_1A_2}{\rho}\right)_j\to\left(\frac{D_1A_2}{\rho}\right)_\wedge,
    \qquad \left(\frac{D_2A_1}{\rho}\right)_j\to\left(\frac{D_2A_1}{\rho}\right)_\wedge
    \\ \text{and}\qquad \left(\frac{D_r X^m}{\rho}\right)_j\to \left(\frac{D_r X^m}{\rho}\right)_\wedge
\end{gather*}
in the norm of $L^2(\Sigma)$.

\noindent \underline{Case A}. Suppose that \eqref{poss1} holds true. Then from (\ref{inposs})
\[
\frac{\langle (W_1)_\wedge,(W_2)_\wedge \rangle^2}{\|(W_1)_\wedge\|^2\|(W_2)_\wedge \|^2}=\frac{b^2(\hat{\phi})}{a^2(\hat{\phi})}=1.
\]
Hence $(W_1)_\wedge$ and $(W_2)_\wedge$ are parallel as elements of
$L^2(\Sigma)$, since $b<0$ we have
\begin{align}
  \left\langle\left(\frac{D_2A_1}{\rho}\right)_{\wedge}^2\right\rangle&=0 \label{ide1} \\
   \left\langle\left(\frac{\{X^m,X^n\}}{A}\right)_{\wedge}^2\right\rangle&=0 \label{ide2} \\
    \left\langle\left[\left(\frac{D_1A_2}{\rho}\right)_{\wedge}+ \left(\frac{\{A_1,A_2\}}{A}\right)_{\wedge}\right]^2\right\rangle&=0 \label{ide3} \\
    \left\langle\left[\left(\frac{D_1X^m}{\rho}\right)_{\wedge}+ \left(\frac{ \{A_1,X^m\}}{A}\right)_{\wedge}\right]^2\right\rangle& =0 \label{ide4} \\
    \left\langle\left[\left(\frac{D_2X^m}{\rho}\right)_{\wedge}+ \left(\frac{\{A_2,X^m\}}{A}\right)_{\wedge}\right]^2\right\rangle&=0 \label{ide5}
\end{align}
All the fields involved in the above are a linear
combination of finitely many harmonic functions and hence they are continuous
as maps of $\Sigma$. Thus the terms inside the integral symbols on the left side of \eqref{ide1}-\eqref{ide5} also vanish pointwise.
It is readily seen that $\left(\frac{D_2A_1}{\rho}\right)_\wedge =0$. Moreover, since
\[
   \left(\frac{D_1A_2}{\rho}\right)_{\wedge}+ \left(\frac{\{A_1,A_2\}}{A}\right)_{\wedge}
   =\left(\frac{D_1A_2}{\rho}\right)_\wedge\left(1-\frac{D_2A_1}{\rho a}\right)_\wedge
\]
where $a(\hat{\phi})\not=0$, \eqref{ide3} yields $\left(\frac{D_1A_2}{\rho}\right)_\wedge =0$.
Since
\[
   \left(\frac{D_1X^m}{\rho}\right)_{\wedge}+ \left(\frac{\{A_1,X^m\}}{A}\right)_{\wedge}=
   \left(\frac{D_1X^m}{\rho}\right)_\wedge\left(1-\frac{D_2A_1}{\rho a}\right)_\wedge,
\]
\eqref{ide4} yields $\left(\frac{D_1 X^m}{\rho}\right)_\wedge=0$. Since
\[
   \left(\frac{D_2X^m}{\rho}\right)_{\wedge}+ \left(\frac{\{A_2,X^m\}}{A}\right)_{\wedge}=
   \left(\frac{D_2X^m}{\rho}\right)_\wedge\left(1+\frac{D_1A_2}{\rho a}\right)_\wedge-\left(\frac{D_2A_2D_1X^m}{a\rho^2}\right)_{\wedge},
\]
the last term vanishes
since,
\[
\frac{D_2A_2}{\rho}\vert_{\wedge}=D_2\left(\frac{A_2}{\rho}\right)_{\wedge},
\]

\[
\frac{D_1A_2}{\rho}\vert_{\wedge}=0 \to \frac{A_2}{\rho}\vert_{\wedge}=0
\]
 and also $\left(\frac{D_1X^m}{\rho}\right)_{\wedge}=0,$ then
\eqref{ide5} yields $\left(\frac{D_2 X^m}{\rho}\right)_\wedge=0$. This is a contradiction since
\begin{equation} \label{contr}
\left<
\left(\frac{D_2A_1}{\rho}\right)^2+\left(\frac{D_1A_2}{\rho}\right)^2+\left(\frac{D_1X^m}{\rho}\right)^2
+\left(\frac{D_2X^m}{\rho}\right)^2\right>=1\not=0
\end{equation}
for all $\phi\in S^{16d-1}$.

\noindent \underline{Case B}. Suppose now that \eqref{poss2} holds true. Let
\[
    F_j=\left(\frac{D_2 A_1}{\rho a}\right)_j\qquad \text{and} \qquad
    G_j=\left(\frac{D_1 A_2}{\rho a}\right)_j.
\]
Both these fields are independent of $\rho$. Since
\[
     \frac{\langle (W_1)_j,(W_2)_j \rangle^2}{\|(W_1)_j\|^2\|(W_2)_j \|^2}\to 1
\]
and all the expressions on the left hand side below are continuous in the norm of $L^2(\Sigma)$,
then
\begin{align}
  \left(\frac{D_2A_1}{\rho}\right)_{j}&\to0 \label{idej1} \\
   \left(\frac{\{X^m,X^n\}}{A}\right)_{j}&\to0 \label{idej2} \\
    \left(\frac{D_1A_2}{\rho}\right)_{j}+ \left(\frac{\{A_1,A_2\}}{A}\right)_{j}&\to0 \label{idej3} \\
    \left(\frac{D_1X^m}{\rho}\right)_{j}+ \left(\frac{ \{A_1,X^m\}}{A}\right)_{j}& \to0 \label{idej4} \\
    \left(\frac{D_2X^m}{\rho}\right)_{j}+ \left(\frac{\{A_2,X^m\}}{A}\right)_{j}&\to0 \label{idej5}
\end{align}
as $j\to+\infty$. The limits in \eqref{idej1}-\eqref{idej5} are regarded in the sense of $L^2$.

The condition \eqref{idej1} and an argument involving continuity in $L^2(\Sigma)$,  imply
\eqref{ide1} and thus $\left(\frac{D_2A_1}{\rho}\right)_\wedge=0$
pointwise as in the previous case. We now show that also
$\left(\frac{D_1A_2}{\rho}\right)_\wedge=0$ under condition
\eqref{poss2} as follows. Since every field $-\N\cup\N$ in $\cH^1$
is orthogonal to the constant fields and the total derivative $D_2$
leaves invariant $\Rspan\{Y^A:A\in
\N\}$,  $\langle
(1-F_j)^2\rangle \geq 1$ for all $j\in \NN$. By virtue of the
compactness of $S^{D-1}$, and after the extraction of a subsequence
if necessary, we can then assume that

\[
       \frac{(1-F_j)^2}{\|(1- F_j)^2\|}\to F_{-}
\]
for suitable $F_-\in L^2(\Sigma)$ such that $\|F_- \|=1$.
Since
\[
   \left\langle \left[\left(\frac{D_1A_2}{\rho}\right)_{j}+ \left(\frac{\{A_1,A_2\}}{A}\right)_{j}\right]^2\right\rangle =
   \left\langle \left(\frac{D_1A_2}{\rho}\right)^2_j\left(1-F_j\right)^2\right\rangle,
\]
\eqref{idej3} implies
\[
   \|(1-F_j)^2\| \left\langle \left(\frac{D_1A_2}{\rho}\right)_j^2\frac{\left(1-F_j\right)^2}{\|(1- F_j)^2\|}\right\rangle=
   \left\langle \left(\frac{D_1A_2}{\rho}\right)_j^2\left(1-F_j\right)^2\right\rangle\to 0.
\]
Then, since  $ \|(1-F_j)^2\| >C>0$,
\[
    \left\langle \left(\frac{D_1A_2}{\rho}\right)_j^2\frac{\left(1-F_j\right)^2}{\|(1-F_j)^2\|}\right\rangle \to 0.
\]
Thus $\langle \left(\frac{D_1A_2}{\rho}\right)_\Lambda^2 F_{-}
\rangle =0$. Since $F_{-}>0$ in a set of positive measure,
Lemma~\ref{harmonic} applied to $A_2$ and the latter, ensure
$\left(\frac{D_1A_2}{\rho}\right)_\wedge=0$.

In a similar fashion,
\[
   \left\langle \left[\left(\frac{D_1X^m}{\rho}\right)_{j}+
   \left(\frac{\{A_1,X^m\}}{A}\right)_{j}\right]^2\right\rangle=\left\langle\left(\frac{D_1X^m}{\rho}\right)_j^2\left(1-F_j\right)^2\right\rangle,
\]
equation \eqref{idej4}, the fact that $\|F_-\|=1$ and Lemma~\ref{harmonic}, yield $\left(\frac{D_1X^m}{\rho}\right)_\wedge=0$.
Moreover, just as for $F_j$, also $\langle (1+G_j)^2\rangle\geq 1$ for all $j\in \NN$ and we can assume that
\[
       \frac{(1+G_j)}{\|(1+ G_j)\|}\to G_+
\]
for suitable $G_+ \in L^2(\Sigma)$ such that $\|G_+\|=1$. Hence
\begin{equation}
\begin{aligned}
   &\left\langle \left[\left(\frac{D_2X^m}{\rho}\right)_{j}+ \left(\frac{\{A_2,X^m\}}{A}\right)_{j}\right]^2\right\rangle=\\ \nonumber &
 \left\langle \left[\left(\frac{D_2X^m}{\rho}\right)_{j}\frac{(1+G_j)}{\|(1+ G_j)\|}
 -\frac{\left(\frac{D_2A_2}{a\rho}\right)}{\|(1+G_j\|}\left(\frac{D_1X^m}{\rho}\right)_j\right]^2\right\rangle\to 0,
 \end{aligned}
\end{equation}
implying $\left(\frac{D_2X^m}{\rho}\right)_{\Lambda}=0$,
since \[
\frac{\left(\frac{D_2A_2}{a\rho}\right)}{\|(1+G_j\|}
\] is bounded.
In fact
\[
\|D_2A_2\|^2=\sum_{m,n}^{N}m^2\vert A_{2(m,n)}\vert^2\le N^2 \sum_{m,n}^{N}\vert A_{2(m,n)}\vert^2\le N^2 \|D_1A_2\|^2,
\]
since
\[
\|D_1 A_2\|^2=\sum_{m,n}^N n^2\vert A_{2(m,n)}\vert^2 \quad
\textit{implies}\quad \sum_{m,n}^N\vert A_{2(m,n)}\vert^2\le
\|D_1A_2\|^2.
\]
Hence
\[
\left(\frac{(\frac{D_2A_2}{a\rho})}{\|(1+ G_j)\|}\right)^2\le N^2\frac{\|G_j\|^2}{1+\|G_j\|^2}\le N^2.
\]
This shows that \eqref{poss2} contradicts \eqref{contr} and
completes the proof of the theorem.
\end{proof}

Note that $|b|/a=\cos(\theta)$ where $\theta$ is the angle between
$W_1$ and $W_2$. The crucial point for the validity of Theorem~\ref{bound_V_finite_case}
is the fact that there is an uniform lower bound for $\theta$ whenever
$A_r,X^m\in \Rspan\{Y^A:A\in
\N\}$.

We may now use the Lemma~\ref{nuevo_lemma} to show that all these regularized models have discrete spectrum with finite multiplicity. A similar bound  may be obtained for the exact potential of the supermembrane with central charges. By taking the large $N$ limit, with constant $k\ne 0$, this strongly suggests that the exact Hamiltonian has discrete spectrum.

%%%%%%%%%%%%%%%%%%%%%%%%%%%%%%%%%%%%%%%%%%%%%%%%%%%%%%%%%%%%%%%%%%%
\subsection{An example in small dimension} \label{subsection44}
Let us illustrate the previous general proof in a very simple case.
We consider the Hamiltonian of previous section in the particular case in which the basis only contains the elements:
\[ \{Y^{(-1,-1)},Y^{(-1,0)}, Y^{(0,-1)},  Y^{(-1,1)}, Y^{(1,-1)},Y^{(1,0)}, Y^{(0,1)},Y^{(1,1)}\}\] where $Y^{(m,n)*}= Y^{(-m,-n)}$, with $Y^{(m,n)}=e^{i(m\widehat{X}_1+n\widehat{X}_2)}$.  This model satisfies the bound obtained in Theorem~\ref{bound_V_finite_case}, however we illustrate the abstract argument in this simple setting.

In the gauge fixing condition  we consider the fields $A_{rs}$, satisfying
\begin{equation}
\begin{aligned}
 A_1&=A_{(1,0)}Y^{(1,0)}+c_1\\
 A_2&= B_{(0,1)}Y^{(0,1)}+B_{(1,1)}Y^{(1,1)}+B_{(-1,1)}Y^{(-1,1)}+c_2
\end{aligned}
\end{equation}
where $c_i$ are coupling constants, $B_{(m,n)}^*=B_{(-m,-n)}$ and $A_{(m,0)}^*=A_{-m,0}$. According to our definition, the covariant derivatives act on the gauge fields as follows:
\[
D_1A_2=inB_{(m,n)}Y^{(m,n)}, D_2A_1=-imA_{(m,0)}Y^{(m,0)}.
\]
Let us consider one of the terms of the potential (a similar argument follows for any the others),
\begin{equation}
\begin{aligned}
 V=&\left\langle \left\vert D_1A_2+i\{A_1,A_2\}\right\vert^2+\vert D_2A_1\vert^2\right\rangle\\ &=
\rho^2\left(\left|\frac{iB_{(0,1)}-A_{(1,0)}B_{(-1,1)}+A_{(-1,0)}B_{(1,1)}}{\rho}\right|^2+\left|\frac{iB_{(1,1)}-A_{(1,0)}B_{(0,1)}}{\rho}\right|^2\right)\\ \nonumber &+\rho^2\left(\left|\frac{iB_{(1-,1)}+A_{(1,0)}B_{(0,-1)}}{\rho}\right|^2+
 \left|\frac{A_{(1,0)}B_{(1,1)}}{\rho}\right|^2+\left|\frac{A_{(-1,0)}B_{(-1,1)}}{\rho}\right|^2+\left|\frac{A_{(1,0)}}{\rho}\right|^2
\right),
\end{aligned}
\end{equation}
where $\rho^2=\vert A_{(1,0)}\vert^2+ \vert B_{(0,1)}\vert^2+ \vert B_{(1,1)}\vert^2+\vert B_{(-1,1)}\vert^2.$
We establish the lower bound found in Theorem~\ref{bound_V_finite_case} in this particular setting by mimicking the abstract proof.

To this end, it is enough to show that the sum of all brackets is  bounded below by $C>0$.
Firstly we check that there are no points at which all brackets vanish.
We have:\begin{equation}
\begin{aligned}
  \frac{1}{\rho}(iB_{(0,1)}-A_{(1,0)}B_{(-1,1)}+A_{(-1,0)}B_{(1,1)})=0 \label{i1}
  \end{aligned}
    \end{equation}
 \begin{equation}
  \begin{aligned}
    \frac{1}{\rho}(iB_{(1,1)}- A_{(1,0)}B_{(0,1)})=0 \label{i2}
    \end{aligned}
    \end{equation}
    \begin{equation}
\begin{aligned}
     \frac{1}{\rho}(iB_{(-1,1)}+A_{(-1,0)}B_{(0,1)})=0\label{i3}
     \end{aligned}
    \end{equation}
     \begin{equation}
\begin{aligned}
     \frac{1}{\rho}A_{(1,0)}B_{(1,1)}=0\label{i4}
     \end{aligned}
    \end{equation}
     \begin{equation}
\begin{aligned}
     \frac{1}{\rho}A_{(1,0)}B_{(1,-1)}=0\label{i5}\\
     \end{aligned}
    \end{equation}
     \begin{equation}
\begin{aligned}
     \frac{A_{(1,0)}}{\rho}=0\label{i6}
    \end{aligned}
    \end{equation}
These identities are equivalent to  the following system:
  \begin{equation*}
\begin{pmatrix}
i & A_{(-1,0)}& -A_{(1,0)}\\
-A_{(1,0)} & i & 0 \\
A_{(-1,0)}& 0 & i\\
\end{pmatrix} \begin{pmatrix}  \frac{B_{(0,1)}}{\rho}\\
 \frac{B_{(1,1)}}{\rho} \\
\frac{B_{(-1,1)}}{\rho} \\ \end{pmatrix} =0.
\end{equation*}
Solving the latter, gives $\vert A_{(1,0)}\vert=\frac{1}{\sqrt{2}}$. By virtue of \eqref{i4} and  \eqref{i5}, this ensures $\frac{B_{(1,1)}}{\rho}=0$ and $\frac{B_ {(-1,1)}}{\rho}=0$. Moreover, from  \eqref{i1} we find that $\frac{B_{(0,1)}}{\rho}=0$. The latter is impossible, hence there is no solution to \eqref{i1}-\eqref{i6}.

Denote by $f$ the sum of the brackets in the expression for the potential $V$. This is a real continuous function of  the real and imaginary part of $B_{(m,n)}$ and $A_{(1,0)}$. Let us examine the compact set $\vert A_{(1,0)}\vert \le R_0 $. Since the other variables are divided by $\rho$, they are defined on a compact set. In this compact set, $f$ has a minimum value which must be different from zero,
\[
f\ge C>0.
\]

Consider now the case when $\vert A_{(1,0)}\vert \to +\infty$, which should also render $f\ge C >0$. Otherwise, there would be a sequence on which $f\to 0$. However from the expression for $f$,
\[\left\vert\frac{B_{(1,1)A_{(1,0)}}}{\rho}\right\vert \to 0\qquad \text{and}
\qquad\left\vert \frac{B_{(-1,1)}A_{(1,0)}}{\rho}\right\vert\to 0.\] Hence \[\left\vert  \frac{B_{(1,1)}}{\rho}\right\vert\to 0\quad \textrm{and}\quad \left\vert \frac{B_{(-1,1)}}{\rho}\right\vert \to 0,\] and we can also conclude that $\left\vert \frac{B_{(0,1)}}{\rho}\right\vert\to 0$. By looking at the last term of $f$ we also have $\vert \frac{A_{(1,0)}}{\rho}\vert\to 0$. Again this is impossible, $f$ must thus be bounded from below by $C>0$.

By a direct application of Lemma~\ref{nuevo_lemma}, the regularized model incorporating the fermionic contribution has a purely discrete spectrum.

\subsection{SU(N) regularization} \label{subsection46}
This regularization of the supermembrane with central charges was proposed in \cite{gmr}. It is invariant under infinitesimal transformations generated by the first class constraint obtained by variations on $\Lambda$ of the Hamiltonian below. This first class constraint satisfies an $SU(N)$ algebra.
 The
resulting Hamiltonian proposed in \cite{gmr} has the form
\begin{equation}
\begin{aligned}
H= & \mathrm{Tr}\left(\frac{1}{2N^{3}}((P_{m})^2+ (\Pi_{r})^{2})+
\right. \nonumber \\& +\frac{n^2}{16\pi^2N^3}(i[X^{m},X^{n}])^2+
\frac{n^2}{8\pi^2N^3}\left(\frac{1}{N}[T_{V_{r}},X^{m}]T_{-V_{r}}+i
[\mathcal{A}_r,X^{m}]\right)^2+\nonumber \\&
+\frac{n^2}{16\pi^2N^3}\left(i[\mathcal{A}_r,\mathcal{A}_s]-
\frac{1}{N}([T_{V_s},\mathcal{A}_r]T_{-V_s}-[T_{V_r},\mathcal{A}_s]
T_{-V_r})\right)^2 + \frac{1}{8}n^2+ \nonumber \\& +\frac{n}{4\pi N^3}
  \Lambda\left([ X^{m},P_{m}]- \frac{i}{N}[T_{V_r},\Pi_{r}]T_{-V_r}
  +[ \mathcal{A}_{r},\Pi_{r}]\right)+ \nonumber \\
   &+ \frac{in}{4\pi N^3}(\overline{\Psi}\Gamma_{-}\Gamma_{m}
   \lbrack{X^{m},\Psi}\rbrack
   -\overline{\Psi}\Gamma_{-}\Gamma_{r}\lbrack{\mathcal{A}_{r},\Psi}
   \rbrack +
  \Lambda \lbrack{\overline{\Psi}\Gamma_{-},\Psi}\rbrack + \nonumber \\
  & - \left.\frac{i}{N} \overline{\Psi}\Gamma_{-}\Gamma_{r}
  [T_{V_{r}},\Psi] T_{-V_r})\right)  \label{e12}
\end{aligned}
\end{equation}
Here $A=(a_1,a_2)$, where the indices $a_1,\,a_2=0,\ldots,N-1$ exclude
the pair $(0,0)$, $V_1=(0,1)$, $V_2=(1,0)$ and $T_0\equiv T_{(0,0)}=N\,
\mathbb{I}$. We agree in the following convention
\begin{equation*}
\begin{gathered}
 X^{m}=  X^{mA}T_{A},\quad\qquad P_{m}=P^{A}_m T_{A},\\
 \mathcal{A}_r= \mathcal{A}^A_r T_{A}, \quad\qquad
 \Pi_{r}=\Pi_r^{A}T_{A},
\end{gathered}
\end{equation*}
where $T_A$ are the generators of the $SU(N)$ algebra:
\begin{equation*}
 [T_{A},T_{B}]= f^{C}_{AB}T_{C}.
\end{equation*}

 We may use the area preserving symmetry  of the supermembrane which
reduces to a $SU(N)$ gauge symmetry on the regularized model to fix a particular gauge. We are
allowed to consider the $A_r, r=1,2$ with the following expressions
\begin{equation}
\begin{aligned}
A_1=A^{(m,0)}T_{(m,0)}\quad m\neq 0\\
A_2=A^{(p,q)}T_{(p,q)}\quad p\neq 0,
\end{aligned}
\end{equation}
with all other components of the gauge fixed to zero.
See \cite{gmr} for details of the gauge fixing procedure. We use the
Heisenberg-Weyl generators $T_{(p,q)}$ to express the $SU(N)$ valued
objects.

We now check that the potential satisfies the hypothesis of
 Lemma~\ref{nuevo_lemma}. The argument follows in analogous fashion as in the previous cases. We may write the potential as
\[V=\frac{n^2}{16\pi^2 N^3}\rho^2 P(\rho,\varphi),
\] where
\[P(\rho,\varphi)=1+2b\rho+a^2\rho^2=\left(a\rho+\frac{b}{a}\right)^2+\left(1-(\frac{b}{a})^2\right)
\] denoting
\begin{equation}
\begin{aligned}
&\rho=\sum_{m,r}\Tr\left( \frac{1}{N}[T_{V_r},X^m]T_{-V_r}\right)^2+\sum_{s,r}\left(\frac{1}{N}[T_{V_s},A_r]T_{-V_s}\right)^2\\ \nonumber &
\varphi=\left(\frac{X^m}{\rho},\frac{A_r}{\rho}\right).
\end{aligned}
\end{equation}
In the above expressions $a=a(\varphi),b=b(\varphi)$, $a$ is taken positive.
Again as in the previous proofs,  see the original argument in \cite{bgmr2},
$\left(\frac{b}{a}\right)$ is the inner product of two unitary vectors, hence $\left(\frac{b}{a}\right)^2\le 1$.

We assume there is a sequence $\rho_l,\varphi_l$ such that $P(\rho_l,\varphi_l)\to 0$ then we must have
\begin{equation}\label{uno}
\left(a(\varphi_l)\rho_l+\frac{b(\varphi_l)}{a(\varphi_l)}\right)\to
0
\end{equation}
and
\begin{equation}\label{dos}
1-\left(\frac{b(\varphi_l)}{a(\varphi_l)}\right)^2\to 0.
\end{equation}
Since $\varphi_l$ takes values on a compact set, necessarily $\varphi_l\to \hat{\varphi}$ (at least for a
sequence) and we must have
\[\left(\frac{b(\hat{\varphi}_l)}{a(\hat{\varphi}_l)}\right)^2=1.
\]
Note that, even if $a\to 0$, the quotient is always well defined
since it is the inner product of unitary vectors. Again, as in
previous sections, the two vectors must be opposite in order to satisfy
(\ref{uno}). In particular we must have \[\left(
\frac{i}{N}
\left[T_{V_r},(\frac{X^m}{\rho})_l\right]T_{-V_r}-\left[\left(\frac{A_r}{\rho}\right)_l,\left(\frac{X^m}{\rho}\right)_l\right]\right) \frac{1}{a_l}\to 0.
\]
Consequently, multiplying by $(\frac{X^m}{\rho})_l$ and taking the trace,
\[
\Tr\left(\left[T_{V_r},(\frac{X^m}{\rho})_l\right]T_{-V_r}\left(\frac{X^m}{\rho}\right)_l\right)\to
0.
\]
On the other hand , however, $\left(\frac{X^m}{\rho}\right)$ converges to
$\left(\frac{X^m}{\rho}\right)_\Lambda$ and the above is a
continuous function. Hence we must
have
\[
\Tr\left(\left[T_{V_r},(\frac{X^m}{\rho})_\Lambda\right]T_{-V_r}\left(\frac{X^m}{\rho}\right)_\Lambda\right)=0
\]
which implies \[ \left(\frac{X^m}{\rho}\right)_\Lambda=0.
\]

The curvature term in the potential  (with the current fixed gauge)  splits into two squared terms:
\[
\Tr\left\vert
\frac{1}{N}[T_{V_2},A_1]T_{-V_2}\right\vert^2+\Tr\left(\frac{1}{N}[T_{V_1},A_2]T_{-V_1}+i[A_1,A_2]
\right)^2.
\]
By applying the previous argument, the condition $P\to 0$ implies
also \[\frac{A_r}{\rho}\to 0.\] The latter is a contradiction, since
it is inconsistent with the definition of $\rho$. We then conclude that the assumption
$P\to 0$ is impossible and consequently $P(\rho,\varphi)>C>0$.

The
hypothesis of Lemma~\ref{nuevo_lemma} is then satisfied, and the regularized $SU(N)$
model including the fermionic terms has a purely discrete spectrum. The bosonic potential of the theory in the continuum was shown in \cite{bgmr2} to satisfy the same type of bound as given in Lemma~\ref{nuevo_lemma}. It
 seems reasonable to expect that in the large $N$ limit, the $SU(N)$ matrix model will converge to the one of the supermembrane with central
charges in the continuum. In this respect, note that the spectrum of the
semiclassical Hamiltonian obtained from this regularized model also converges to
the spectrum of the semiclassical supermembrane with central charge
in the following sense. Take a value of energy $E$ and consider the
eigenvalues of the $SU(N)$ which are below $E$: $\lambda_N^i<E$.
Find the large $N$ limit of $\lambda_N^i$ for a fixed $E$. The
limit of $\lambda_N^i$ is exactly the set of eigenvalues of the
semiclassical supermembrane with central charges satisfying the $E$
bound. This can be explicitly evaluated, \cite{bgmr2}.

%%%%%%%%%%%%%%%%%%%%%%%%%%%%%%%%%%%%%%%%%%%%%%%%%%%%%%%%%%%%%%%%%%%%%%%%%%%%%%%%%%%%%%%%%%%
%%%%%%%%%%%%%%%%%%%%%%%%%%%%%%%%%%%%%%%%%%%%%%%%%%%%%%%%%%%%%%%%%%%%%%%%%%%%%%%%%%%%%%%%%%%%%%%%%%%%%%%%%%%%%%%%%%%%

\section{The D2-D0 model} \label{section5}
We now consider a model which describes the reduction of a 10D $U(N)$ Super Yang-Mills to (1+0) dimensions \cite{witten}, allowing the presence of monopoles. Consider the (2+1) bosonic Hamiltonian given by
\[
H_B=\int_{\Sigma}\Tr\left[ \frac{1}{2}((P^m)^2+(\Pi^i)^2)+\frac{1}{4}\left(F_{ij}^2+2(D_i X^m)^2+(i[X^m,X^n])^2\right)\right]
\]
which satisfies the monopole condition
\[
\int_{\Sigma}\Tr F=2\pi m,\quad m\in\mathbb{N}.
\]
Factoring the $U(N)$-valued 1-form  $\hat{A}$ as $\hat{A}=aD+A$, with $A\in SU(N)$
 and $D\in U(N)$, the monopole condition becomes \[
\int_{\Sigma} \ud a=2\pi\frac{m}{\Tr D}.\] We then write the $1+0$ Hamiltonian as $H=\frac{1}{2}\tilde{H}$,
where
\[
\tilde{H}=-\Delta +V_B+V_F
\]
for
\begin{align*}
V_B &
=\frac{1}{2}\Tr \left[ (i[X^m,X^n])^2+2(i[X^m,\hat{A}_i])^2+F_{ij}^2 \right], \\
 F_{ij} &= \frac{mD\epsilon_{ij}}{\Tr D}+ i[\hat{A}_i,\hat{A}_j]
 \end{align*}
 and $V_F$ is the supersymmetric Yang-Mills fermionic potential.

 In order to characterize the spectrum of $\widetilde{H}$, we observe that there are directions escaping to infinity at which $V_B$ remains finite. In fact, in any direction at which all the brackets vanish, a wave function can escapes to infinity with finite energy. Therefore the spectrum has necessarily a continuous sector. To construct precisely a wave function in the corresponding $L^2$ space, we introduce
 \[ X=\frac{1}{d+2}(\sum_m X^m +\sum_i A_i)\]
 where the index $m$ runs up to the value $d$ and  the range of $i=1,2$. Let \[\widetilde{X}^m=X^m-X
 \qquad \text{and}\qquad \widetilde{A}_i=\hat{A}_i-X.\]
We simplify the argument by taking $D$ to be diagonal. We may then use the gauge freedom of the model to impose that $X$ is also diagonal. The potential is then re-written as
 \begin{align*}
 V_B&
=\frac{1}{2}\Tr\left[(i[\widetilde{X}^m,\widetilde{X}^n])^2+2(i[\widetilde{X}^m,\widetilde{A}_i])^2+(\frac{m D\epsilon_{ij}}{\Tr D}+i[\widetilde{A}_i,\widetilde{A}_j])^2\right] \\
&\qquad +(d+2)\Tr (i[\widetilde{X}^M,X])^2
\end{align*}
where $M=1,\dots,d+2$, $\widetilde{X}^{d+1}=\widetilde{A}_1$ and $\widetilde{X}^{d+2}=\widetilde{A}_2.$
The crucial point is then that the dependence on $X$ is only quadratic.

Following \cite{dwln}, the above allows us to construct a sequence of wave functions
which are singular Weyl sequence for any $E\in [\frac{1}{2}\frac{m^2 \Tr D^2}{(\Tr D)^2},+\infty)$.
These ``pseudo-eigenfunctions'' are the product of a fermionic wave function $\Psi_F$, the bosonic $L^2$ function
  $\phi_0 (\widetilde{X}^m_T, \widetilde{A}^i_T,X)$ and a compactly supported cutoff
  $\chi_t\equiv \chi(\left\|x\right\|-t,\frac{X}{\left\|X\right\|},\widetilde{X}^M_{||})$.
We denote by $B_{||}$  the diagonal part of a matrix $B$ and $B_{T}=B-B_{||}$.

The expressions  for the normalized $\phi_0$ are \[
\phi_0=\left(\left\|X\right\|^l\frac{
\det{g_{ab}}}{\pi^l}\right)^\frac{1}{4}
 \exp\left(-\frac{\left\|x\right\|}{2}(\widetilde{X}^{Ma}g_{ab}\widetilde{X}^{Mb})\right)
\]
\[\int_{\widetilde{X}_{T}^m}\phi_0^2=1\]
where $l=(d+2)(N^2-N)$. Note that $X^{M}=\widetilde{X}^{Ma}T_a$ where $T_a$ are the generators of $U(N)$ and $g_{ab}$ is the square root of the positive symmetric matrix $(d+2)(f_{ab}^c\frac{X^b}{\left\|X \right\|}f_{cd}^e\frac{X^e}{\left\|X \right\|}).$ Note also that $\widetilde{X}^M_{||}$ are singular directions of this matrix for any $\frac{X^b}{\|X\|}$, so they do not appear in the exponential function. Here $\chi(\left\|X \right\|,\widetilde{X}^M_{||})$ has a compact support, hence $\chi(\left\|X \right\|-t,\widetilde{X}^M_{||},\frac{X}{\left\|X \right\|})$ is a wave function with support moving off to infinity as $t\to+\infty$. It is also normalized by the condition
\[ \int_{\widetilde{X}_{||}^m, X} \chi^2_t=1.
\]

The normalized fermionic wave function is the limit when $t\to+\infty$ of the eigenfunction of the fermionic interacting term associated to the negative eigenvalue with largest modulus. One can now evaluate $\lim_{t\to+\infty}(\Psi,H\Psi)$ where $\Psi=\Psi_F\phi_0\chi$. In this limit the only term of $V_B$ that does not vanish is the constant one. There is a cancelation of the quadratic terms in $X$ between the contribution of the Laplacian and that of  the potential. Also the linear term in $\left\|X\right\|$, arising from the action of the Laplacian on $\phi_0$, is exactly canceled by the fermionic eigenvalue which is also linear in $\left\|X\right\|$. This is a supersymmetric effect. Although the monopole in this case breaks supersymmetry, the cancelation occurs exactly as in the model without monopoles. The resulting consequence of this is that
\[
\lim_{t\to+\infty}(\Psi,H\Psi)=\int_{\widetilde{X}_{||},X} \chi_t(-\Delta_x-\Delta_{\widetilde{X_{||}}})\chi_t +\frac{1}{2}m^2\frac{\Tr D^2}{(\Tr D)^2}.
\]
 One may choose $\chi$ such that the first term is equal to any scalar $E\in[0,+\infty)$. Therefore the spectrum of the original Hamiltonian has a continuous part, comprising the interval $[\frac{1}{2}\frac{m^2 \Tr D^2}{(\Tr D)^2},+\infty).$

%%%%%%%%%%%%%%%%%%%%%%%%%%%%%%%%%%%%%%%%%%%%%%%%%%%%%%%%%%%%%%%%%%%%%%%%%%%%%%%%%%%%%%%%%%%%%%%%%%%%%%%%%%%%%%
%%%%%%%%%%%%%%%%%%%%%%%%%%%%%%%%%%%%%%%%%%%%%%%%%%%%%%%%

\section{Models with non-empty essential spectrum} \label{section6}
In this final section we examine in detail three toy models. The common feature of these three models is the fact that they posses an infinite interval of continuous spectrum. In the first example, we start from what we call the dWLN toy model with two flat directions. We then  remove one of the flat directions to get a second example. We remove both flat directions for a third example. The latter has no flat direction, but the potential becomes finite at some point at infinity. Its semiclassical approximation, including the fermionic terms, has discrete spectrum but the full Hamiltonian has continuous spectrum. It also has a bound state below the bottom of the essential spectrum.  }

\subsection{The dWLN toy model} \label{subsection61} We firstly consider a canonical toy model which was examined in some detail in \cite{dwln}. Arguably this model resembles the spectral properties of the 11D supermembrane. This turns out to be a example in which a supersymmetric matrix model does not satisfy the conditions of Lemma~\ref{nuevo_lemma}.

Let
\begin{equation}\label{matrix}
H=\begin{pmatrix}
-\Delta +x^2y^2
 & x+iy \\
x-iy&-\Delta +x^2y^2
\end{pmatrix}.
\end{equation}
The bosonic contribution from the potential has a purely discrete
spectrum, although classically the system is unstable along the flat directions
associated to $(0,y)$ and $(x,0)$. The eigenvalues of the supersymmetric potential are of the form
\[\lambda_{\Susy}^\pm(x,y)=x^2y^2\pm\sqrt{x^2+y^2}.\]
 In the directions $(0,y)$ and $(x,0)$,  $\lambda_{\Susy}^{-}(x,y)\to-\infty$. This originates a non-empty  continuous spectrum for \eqref{matrix}.

Singular Weyl sequences can be constructed for each $E\in[0,+\infty)$, \cite{dwln}. As the operator is non-negative, the spectrum is continuous comprising the interval $[0,+\infty)$. This argument does not exclude the existence, for example at $E=0$, of an eigenfunction. See the numerical computation also done in \cite{korcyl}. In \cite{ghh} it was established that there does not exists such eigenfunction for this toy model (see also Figure~\ref{fig5} and the corresponding discussion in Section~\ref{subsection62}). For a detailed characterization of this model see \cite{lundholm}.

\subsection{A toy model with a mass term in one direction} \label{subsection62}
In the following supermembrane toy model, one flat direction of the previous one has been eliminated but not both. Let
\begin{equation}\label{matrix7}
H=\begin{pmatrix}
-\Delta+y^2+x^2y^2
 & x+iy \\
x-iy& -\Delta+y^2+x^2y^2
\end{pmatrix}.
\end{equation}
Here $H=\frac{1}{2}\{Q,Q^{\dag}\}$ for
\begin{equation}\label{matrix9}
Q =\begin{pmatrix}
  -xy +iy & i\partial_x+ \partial_y\\
i\partial_x-\partial_y & xy-iy
\end{pmatrix}.
\end{equation}
Classically this system has flat directions, since $V_B$  vanishes at $(x,0)$.
The bosonic Hamiltonian can also be bounded from below by an harmonic oscillator in one of the directions but not in both, so once again this model does not satisfy the hypotheses of Lemma~\ref{nuevo_lemma}.

 The bosonic potential has purely discrete
spectrum, see \cite{amilcar}. This can also be deduced from the following operator bound,
\begin{equation}\begin{aligned}
H_{\mathrm{bosonic}}=&-\Delta +x^2y^2+y^2 \\
&=
\frac{1}{2}(\partial_x^2+\partial_{y}^2)+\frac{1}{2}(\partial_{y}^2+x^2y^2+2y^2)+\frac{1}{2}(\partial_{x}^2+x^2y^2)\\
\nonumber &>\frac{1}{2}(\partial_x^2+\partial_{y}^2+(\sqrt{
x^2+2}+\vert y\vert), \end{aligned}\end{equation}
which implies that $H_{\mathrm{bosonic}}$ has a compact resolvent.
The eigenvalues of the supersymmetric potential are of the form
\[\lambda_{\Susy}^{\pm}(x,y)=x^2y^2+y^2\pm\sqrt{x^2+(y-1)^2}.
\]
Here $\lambda^{-}_{\Susy}(x,y)$ is not bounded from below as in the previous case. Following the arguments of \cite{dwln}, or those presented above, one can show that the spectrum is
continuous comprising the interval $[0,+\infty)$.

Numerical simulations provide an insight on the existence of embedded eigenvalues. In Figure~\ref{fig1} we show the outcomes of the following computational experiment. We have discretized $H$ on test spaces generated by the finite element method, using Hermite polynomials of order 3 on triangular mesh in the region $[-L,L]^2\subset \RR^2$, imposing Dirichlet boundary conditions on the boundary of this box. The conformity of the elements ensures that these test subspaces are contained in the form domain of the corresponding operator $H$. Standard variational arguments ensure that any eigenvalue of the reduction of $H$ in this subspace,
will be an upper bound (counting multiplicity) of spectral point of the original Hamiltonian $H$ as well as its restriction to the box.

The graphs in Figure~\ref{fig1}, from top to bottom, correspond to
numerical approximations of the density function for the ground state for $L=25,\,50,\,100,\,200$. The observable support of this eigenfunction happens to be contained in a thin rectangle near the horizontal axis, so we have only included this region in the pictures. Note that the vertical scaling has been exaggerated.
The quasi-optimal mesh employed to generated these graphs were obtained by an $h$-adaptive procedure.

As $L$ increases, the corresponding ground eigenvalue of $H$ in the box
approaches zero. The numerical evidence strongly suggests that the density function of the eigenfunction has a support that escapes to $+\infty$. This indicate that, perhaps,
there are no embedded eigenvalues also for this model, however this conjecture should be confirmed by further analytical investigation.

\begin{figure}[hh]
\centerline{\includegraphics[height=11cm]{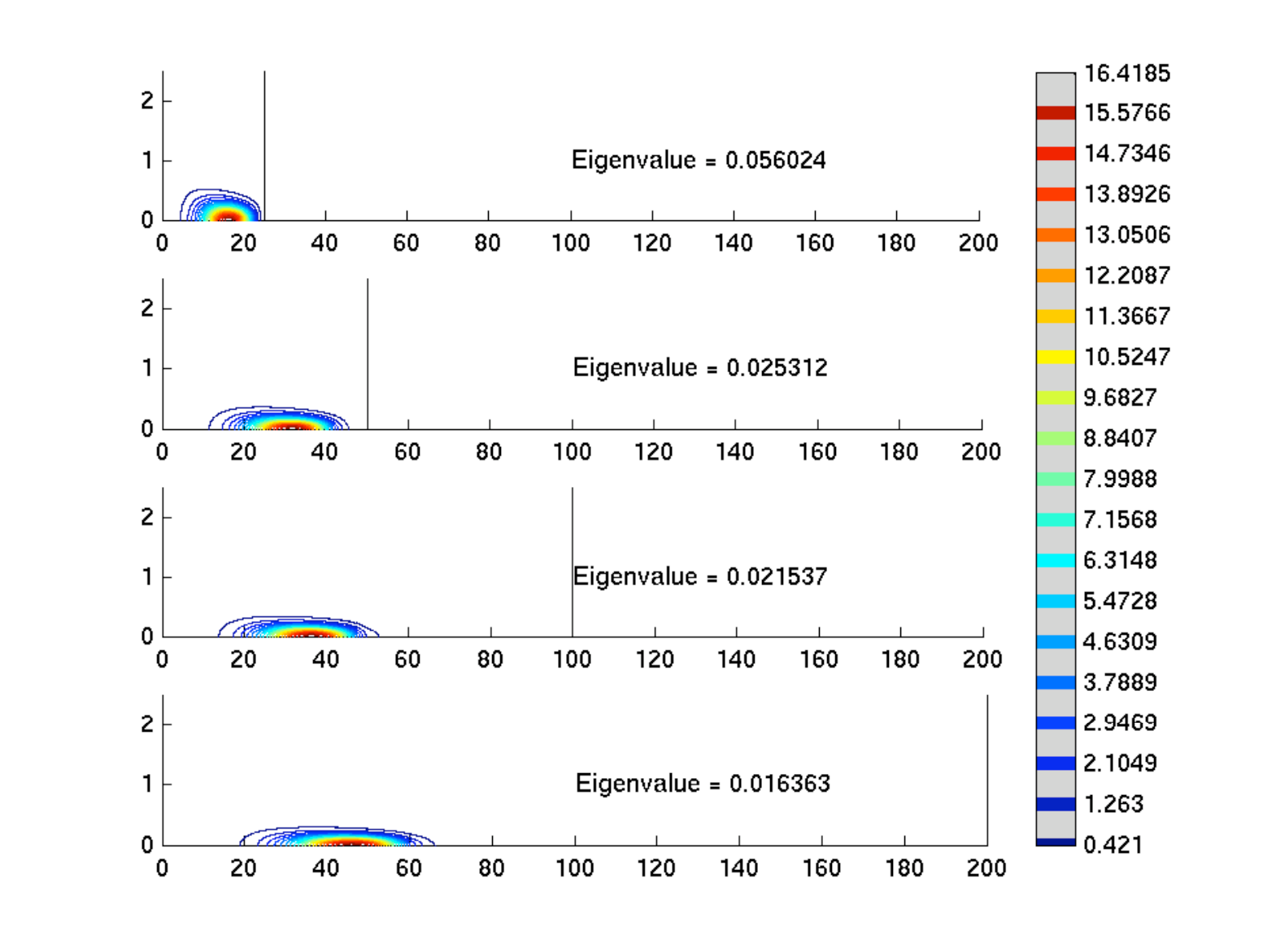}}
\caption{Estimation of the density function for the ground state of $H$ restricted to a box $[-L,L]^2$, subject to Dirichlet boundary conditions, for $L=25\ (\mathrm{top}),\,50,\,100,\,200\ (\mathrm{bottom})$, in the case of the model of Section~\ref{subsection62}. The vertical line depicts part of the boundary of the box. \label{fig1}}
\end{figure}

In order to examine further this statement, in Table~\ref{tab1} we include the numerical estimation of the first 20 eigenvalues of $H$ restricted to the box, for the same test spaces as in Figure~\ref{fig1}. Note that the eigenvalues accumulate in an uniform manner at the origin as $L$ increases. This also indicates that there are no embedded eigenvalues at low energy levels.

\begin{table}[ht!]
\begin{tabular}{c|cccc}
Eigenvalue Number & $L=25$ & $L=50$ & $L=100$ & $L=200$ \\
\hline
1 &  0.0560 &   0.0253 &   0.0215 &   0.0164\\
 &   0.1158 &   0.0466 &   0.0386 &   0.0281\\
 &   0.1970 &   0.0745 &   0.0603 &   0.0426\\
 &   0.2980 &   0.1093 &   0.0852 &   0.0593\\
 5 &  0.4207 &   0.1506 &   0.1097 &   0.0776\\
 \hline
 &   0.5714 &   0.1980 &   0.1372 &   0.0983\\
 &   0.7543 &   0.2515 &   0.1739 &   0.1229\\
 &   0.9701 &   0.3114 &   0.2177 &   0.1511\\
 &   1.2181 &   0.3785 &   0.2661 &   0.1817\\
 10&   1.4974 &   0.4537&    0.3172 &   0.2141\\
 \hline &  1.8056  &  0.5376 &   0.3694   & 0.2492\\
  &  2.1325  &  0.6302 &   0.4249  &  0.2874\\
  &  2.3863  &  0.7313 &   0.4886  &  0.3279\\
  &  2.5953  &  0.8405 &   0.5611  &  0.3700\\
  15&  2.9707 &   0.9576&    0.6401 &   0.4161\\
\hline  &  3.4025 &   1.0828  &  0.7232  &  0.4682\\
  &  3.8690 &   1.2166  &  0.8081  &  0.5257\\
  &  4.3642 &   1.3593  &  0.8633  &  0.5873\\
  &  4.8560 &   1.5104  &  0.8980  &  0.6515\\
  20 &  5.0743  &  1.6687 &   0.9947 &   0.7174
\end{tabular}
\caption{Estimation of the first 20 eigenvalues of $H$ restricted to a box $[-L,L]^2$ subject to Dirichlet boundary conditions in the case of the model of Section~\ref{subsection62}.  \label{tab1}}
\end{table}

For comparison, in Figure~\ref{fig5} we have included a numerical
approximation of the density function of the ground eigenfunction
for the model of Section~\ref{subsection61}.
In this case this density function appears to be localized in a
neighborhood of the axes.
Various other numerical simulations, not presently included,
suggest the following behavior for the ground eigenfunction of the model
Hamiltonian \eqref{matrix}. Unlike the case illustrated in
Figure~\ref{fig1}, the support increases in size,
filling up a ``thin'' T-shaped region close to
$[(-\infty,0)\times 0] \cup [0\times (-\infty,+\infty)]$  as $L\to+\infty$.
In the large $L$ limit, this eigenfunction is not square integrable,
giving rise to a singular Weyl sequence and confirming the absence of an
embedded eigenvalue at the bottom of the spectrum.
This behavior of the approximated eigenfunctions is quite remarkable and it
does not seem to have been noticed before.

\begin{figure}[ht]
\centerline{\includegraphics[height=7cm]{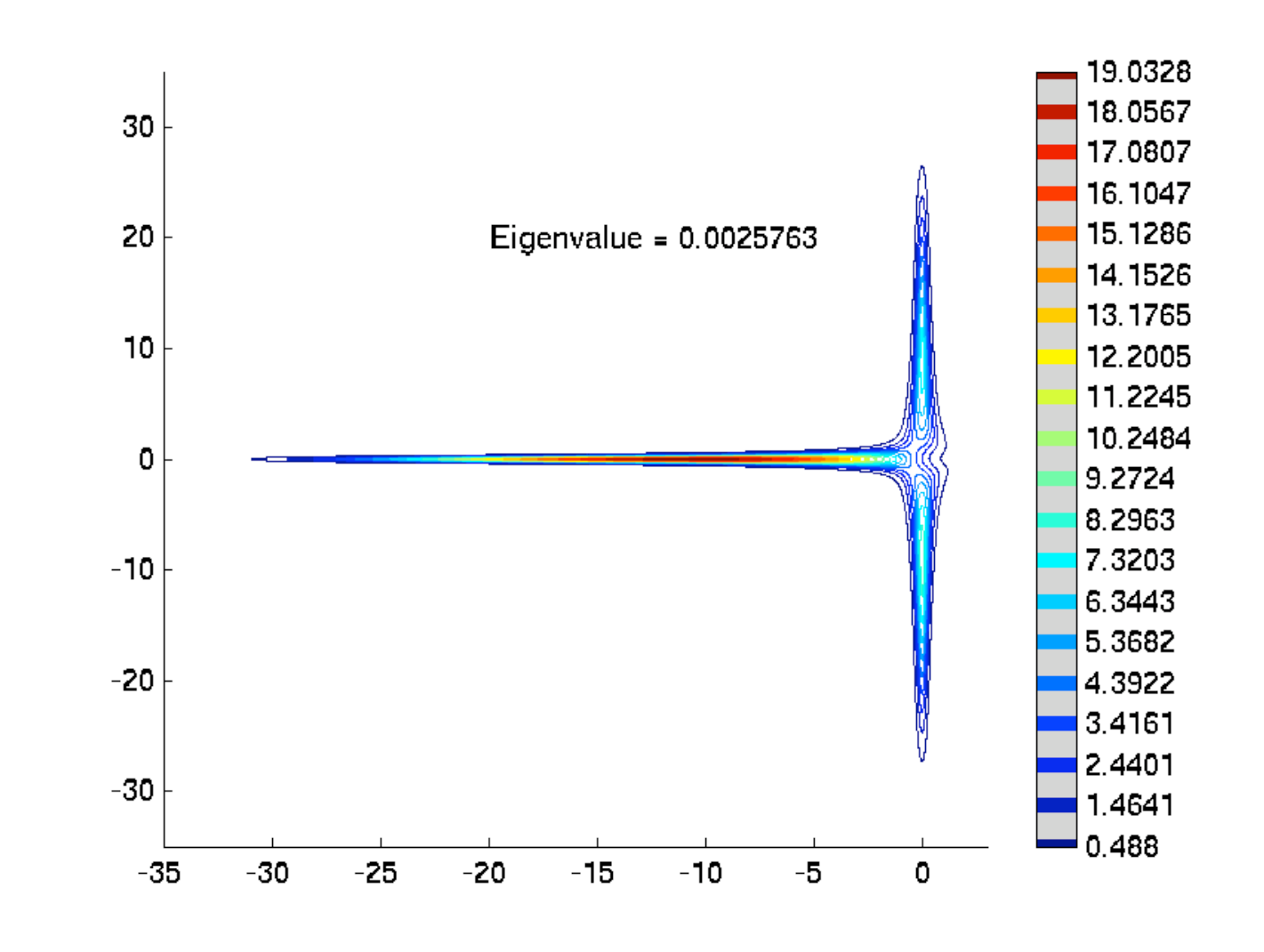}}
\caption{Estimation of the density function for the ground state of $H$ restricted to $[-100,100]^2$ subject to Dirichlet boundary conditions in the case of the model of Section~\ref{subsection61}. The calculations are performed by means of the finite element method and an $h$-adaptive scheme. \label{fig5}}
\end{figure}

\subsection{A toy model with a gap} \label{subsection63}
  Let us now consider the following SUSY Hamiltonian,
\begin{equation*}
H=\begin{pmatrix} -\Delta+V_B(x,y)
 & x+iy+i \\
x-iy-i&-\Delta+V_B(x,y)
\end{pmatrix}
\end{equation*}
where $V_B(x,y)=x^2(y+1)^2+y^2$. This is a supersymmetric quantum mechanical model. In fat, $H=\frac{1}{2}\{Q,Q^{\dag}\}$ where
\begin{equation*}
Q =\begin{pmatrix} -xy-x+iy & i\partial_x+\partial_y\\
i\partial_x-\partial_y & xy+x -iy.
\end{pmatrix}
\end{equation*}
Notice that $H$ is positive therefore it can always be re-written as $H=\hat{Q}\hat{Q}=\frac{1}{2}\{\hat{Q},\hat{Q}\}$ for some self-adjoint operator $\hat{Q}$. Consequently $[H,\hat{Q}]=0$, and $H$ is then the Hamiltonian of a supersymmetric quantum mechanical model.  The same comment is valid for the previous examples.

The bosonic Schr\"ondinger operator $-\Delta + V_B$ has a discrete spectrum, as it satisfies the assumptions of \cite{amilcar}. An independent proof using operator bounds as mentioned above is also possible. The semiclassical approximation, obtained by only considering the quadratic terms in the bosonic potential, satisfies the bound of Lemma~\ref{nuevo_lemma} and hence it has a discrete spectrum with finite multiplicity. The potential $V_B$ alone, does not satisfy this bound directly, because in the direction $(x\to +\infty,y=-1)$ it remains finite. In fact, it turns out to have a non-empty continuous spectrum. The reason for this relies in the behavior of the eigenvalues of the supersymmetric potential,
\[
\lambda_{\Susy}^{\pm}(x,y)=x^2(y+1)^2+y^2\pm\sqrt{x^2+(y+1)^2}.
\]
We observe that $\lambda_{\Susy}^-$ is not bounded from below (although $H\ge 0$) when $y=-1$.

The eigenvalues and eigenfunctions of the fermionic potential
\begin{equation*}
V_F=\begin{pmatrix} -0 & x+iy+i\\
x-iy-i & 0
\end{pmatrix}
\end{equation*}
are \begin{equation*}
\lambda_{\pm}(x,y)=\pm \sqrt{x^2+(y+1)^2}
\qquad \text{and}\qquad \Psi_{\pm}= \frac{1}{\sqrt{2}}\begin{pmatrix}  1\\
\frac{x-i(y+1)}{\lambda_{\pm}}
\end{pmatrix}
\end{equation*}
respectively.
Note that \begin{equation}
\begin{aligned}
\lim_{x\to+\infty}\Psi_{-}= \frac{1}{\sqrt{2}}\begin{pmatrix}  1\\
-1
\end{pmatrix}
\end{aligned}
\end{equation}
and for $x>>(y+1)$ we have $\lambda_{-}=-|x|+O(\frac{1}{|x|}).$ Consider, as in Section~\ref{section5}, the wave function
\[
\Psi=\varphi_0\chi_t \Psi_{-}
\]
where \[
\varphi_0=\left(\frac{|x|}{\pi}\right)^{\frac{1}{4}} \exp\left(-\frac{|x|}{2}(y+1)^2\right),\qquad \chi_t=\chi(|x|-t)\] and the cutoff $\chi(|x|)$ has compact support.
Then
\[
-\frac{\partial^2}{\partial y^2}\varphi_0= -x^2(y+1)^2\varphi_0 +|x|\varphi_0.
\]
The second term on the right hand side exactly cancels the negative contribution from $\lambda_{-}$ and the Laplacian dominates the linear fermionic term. Also for the $y^2$ term, we get
\[  (\Psi, y^2\Psi)\to 1 \quad\textrm{when} \quad t\to+\infty.\]
Finally,
\[
\lim_{t\to+\infty}(\Psi,H\Psi)=1+E \quad \textrm{where} \quad E\in[0,+\infty)
\]
by choosing $\chi$ such that $\int \chi_t (-\frac{\partial^2}{\partial x^2}\chi_t)=E$, which is always possible. Therefore, the interval $[1,+\infty)$ lies in the spectrum of $H$.

Note that there may still be bound states below 1. A similar numerical experiment as the one performed in the previous section confirms this hypothesis, indicating the presence of a ground eigenvalue $\lambda\approx 0.81419$. In Figure~\ref{fig2} we depict a numerical approximation of the density function of the corresponding ground eigenfunction.

\begin{figure}[ht]
\centerline{\includegraphics[height=7cm]{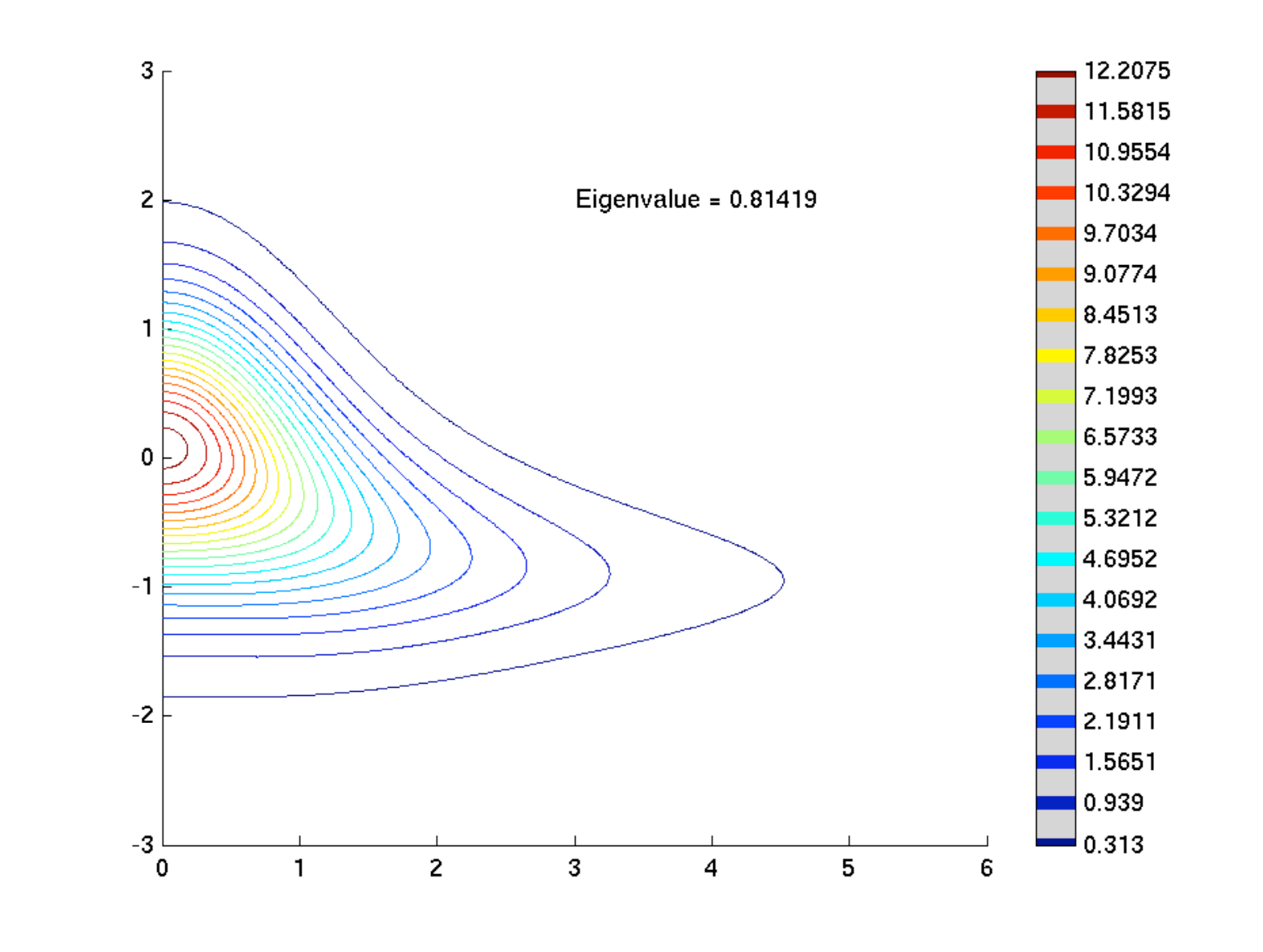}}
\caption{Estimation of the density function for the ground state of $H$ restricted to $[-40,40]^2$ subject to Dirichlet boundary conditions in the case of the model of Section~\ref{subsection63}. The calculations are performed by means of the finite element method and an $h$-adaptive scheme. \label{fig2}}
\end{figure}

In Table~\ref{tab2} we include the numerical estimation of the first 6 eigenvalues of
$H$ restricted to $[-L,L]^2$. This data strongly suggests that the only eigenvalue of $H$ at the bottom of the spectrum is the ground eigenvalue.
This is further confirmed by figures~\ref{fig3} and \ref{fig4}, where we depict
the density functions of the eigenfunctions corresponding to the second and third eigenvalues. The graphs from top to bottom correspond to the values $L=10,\,20,\,40$.
As in the model of Section~\ref{subsection62}, the support of the density functions seem to lie on a narrow strip near the horizontal axis, so we have exaggerated the vertical
scale. Note that in both cases, the eigenfunction is localized in a support that seems to escape to $+\infty$, suggesting no embedded eigenvalue. Note that in this model,  the existence of mass terms by themselves
does not guarantee discreteness of
the supersymmetric spectrum.

\begin{table}[ht!]
\begin{tabular}{c|ccc}
Eigenvalue Number & $L=10$ & $L=20$ & $L=40$ \\
\hline
1&   0.8218 &   0.8142 &   0.8142 \\
2&    1.1937 &   1.0733 &   1.0662 \\
3 &    1.5474 &   1.1707 &   1.1511 \\
4 &    2.0405 &   1.3190 &   1.2803 \\
5 &    2.5313 &    1.5173 &    1.4515 \\
6 &    3.0876 &    1.7618 &    1.6582
\end{tabular}
\caption{Estimation of the first 6 eigenvalues of $H$ restricted to a box $[-L,L]^2$, subject to Dirichlet boundary conditions in the case of the model of Section~\ref{subsection63}. \label{tab2}}
\end{table}

\begin{figure}[ht]
\centerline{\includegraphics[height=9cm]{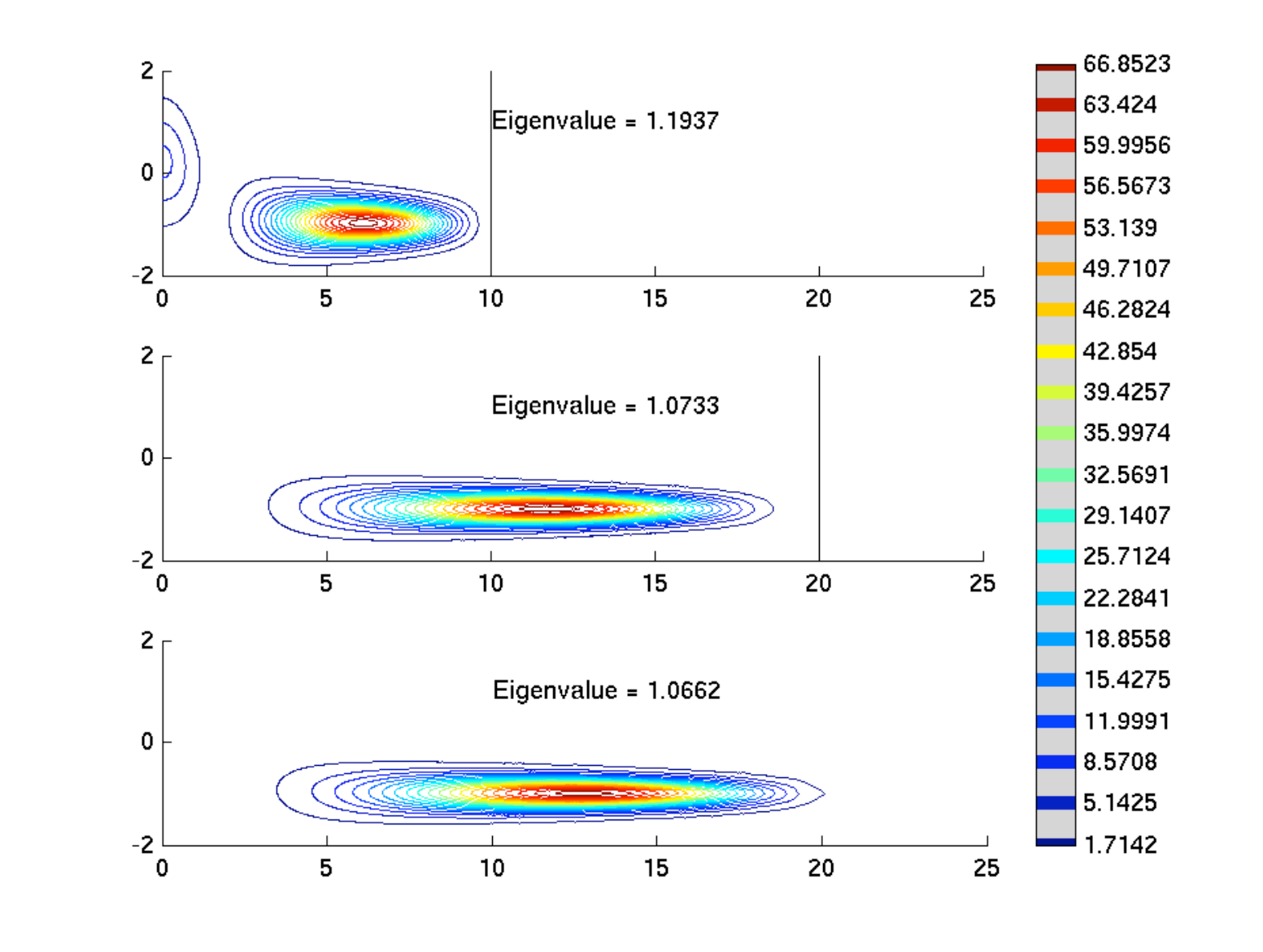}}
\caption{Estimation of the density function for the first excited state of $H$ restricted to a box $[-L,L]^2$, subject to Dirichlet boundary conditions, for $L=10\ (\mathrm{top}),\,20,\,40\ (\mathrm{bottom})$ in the case of the model of Section~\ref{subsection63}. The vertical line depicts part of the boundary of the box. \label{fig3}}
\end{figure}

\begin{figure}[ht]
\centerline{\includegraphics[height=9cm]{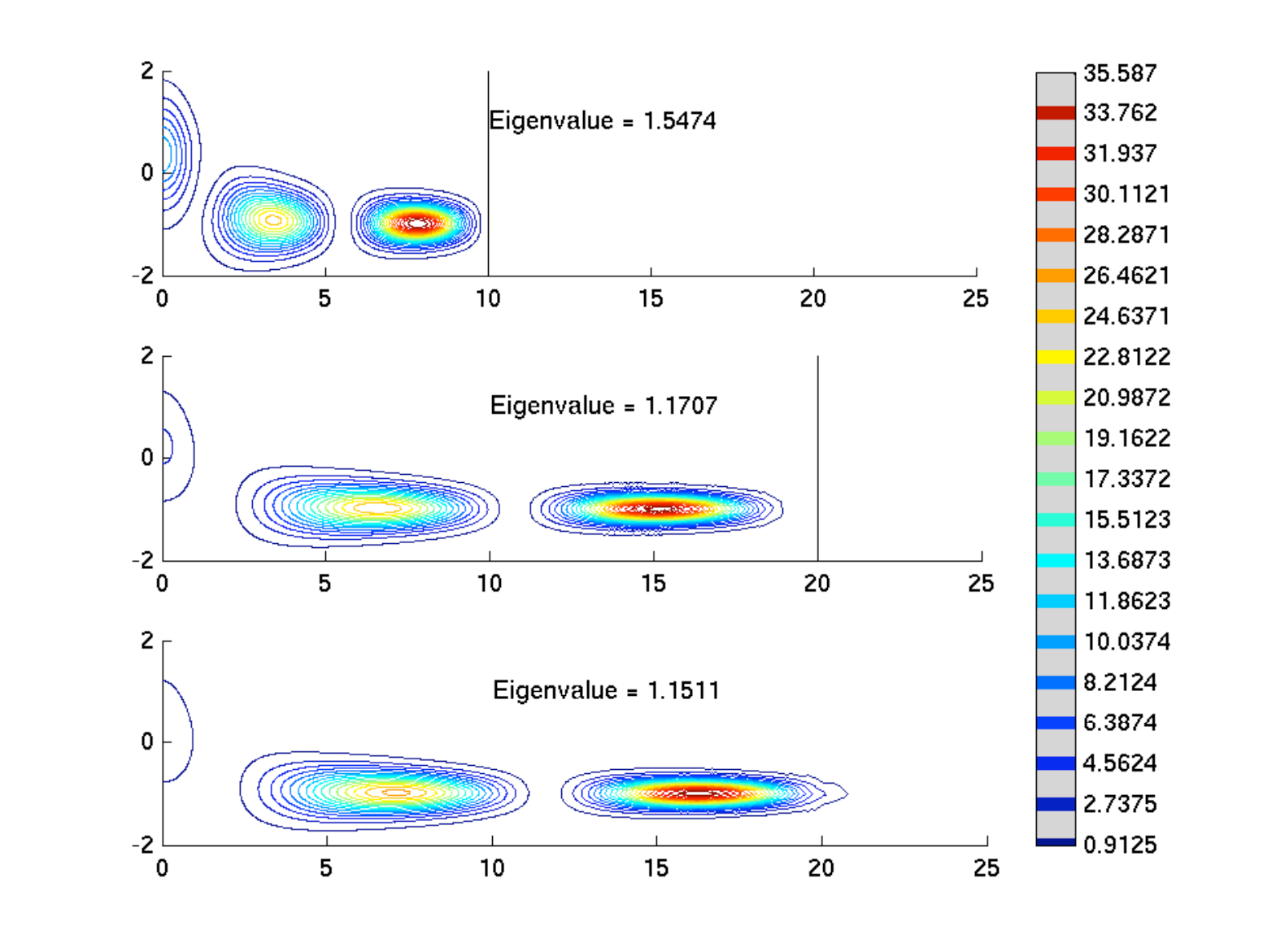}}
\caption{Estimation of the density function for the second excited state of $H$ restricted to a box $[-L,L]^2$, subject to Dirichlet boundary conditions, for $L=10\ (\mathrm{top}),\,20,\,40\ (\mathrm{bottom})$ in the case of the model of Section~\ref{subsection63}. The vertical line depicts part of the boundary of the box. \label{fig4}}
\end{figure}

\section*{Discussion}
 We established  sufficient conditions for spectral discreteness in matrix models. Our criterion apply both to SUSY and non-SUSY models. Those models which satisfy the conditions have a purely discrete spectrum (with finite multiplicity) and only accumulate at $+\infty$. Their resolvent, as well as their associated heat semi-group, are compact.
Mathematically this is an amenable property
as far as the study of the high energy eigenvalues is concerned. On the
one hand, this guarantees the existence of a complete set of
eigenfunctions, which can be used to decompose the action of the operator
in low/high frequency expansions. On the other hand, the study of
eigenvalue asymptotics for the resolvent (or the corresponding heat
kernel) in the vicinity of the origin, can be carried out by means of the
so-called Schatten - von Neumann ideals. None of this extends in
general, if the Hamiltonian has a non-empty essential spectrum.

We proved that the BMN supersymmetric model satisfies sufficiency conditions for discreteness of its spectrum. The bound we find diverges in the large $N$ limit. An open question is still open regarding how to characterize the spectrum in this regime. We conjecture that there will be a non empty essential spectrum and the presence of gaps is not ruled out.  We introduced a top-down regularization for the supermembrane with central charges in addition to the already known SU(N) regularization\cite{gmr, bgmr}. It is possible to show that the operatorial bound for this top-down regularization remains finite for large $N$. In fact the bosonic potential of the exact Hamiltonian with central charge is known to satisfy the bound \cite{bgmr2}, since the bosonic potential of the regularized version converges to the one of the exact theory.  We conclude that the bound should remain finite in the large $N$ limit. This argument gives evidence that  the supermembrane with central charges has a discrete spectrum and consequently could be considered as a fundamental membrane of M-theory.  We demonstrated that the D2-D0 system has a continuous spectrum  as the original (1+0) SYM matrix model. There is also a shift at the bottom of the essential spectrum.

Finally, we examined numerically the ground state of various models including the dWLN toy model. Our simulations indicates that, in the presence of a mass term in one direction, the model preserves its continuous spectrum with no embedded eigenvalue. We also introduced a toy model which has several properties. Firstly, it does not have flat directions. Secondly, its semiclassical approximation has discrete spectrum but the Hamiltonian has continuous spectrum. It also contains a bound state, the ground state, below the bottom of the essential spectrum. It has a gap which we explicitly estimate.

%%%%%%%%%%%%%%%%%%%%%%%

\section*{Acknowledgements}
We thank Y.~Lozano for useful discussions. We are kindly grateful to E.~Witten for motivating discussions at an early stage of this work.
The work of MPGM is funded by
the Spanish ``Ministerio de Ciencia e Innovaci\'on'' (FPA2006-09199) and
the ``Consolider-Ingenio 2010'' programmes CPAN (CSD2007-00042).
AR was supported by the programme ``Campus de la Excelencia de la Universidad de Oviedo''. AR and LB would like to thank the Theoretical Physics Group at the Universidad de Oviedo, where part of this work was carried out, for their kind hospitality and financial support.

% %%%%%%%%%%%%%%%%%%%%%%%%%%%%%%%%%%%%%%%%%%%%%%%%%%%%%%%%%%%%%%%%%%%

\end{document}